\documentclass{article} % For LaTeX2e
\usepackage{iclr2027_conference,times}

\usepackage{booktabs}
\usepackage{amsmath}
\usepackage{amssymb}
\usepackage{mathtools}
\usepackage{amsthm}
\usepackage{wrapfig}
\usepackage[table]{xcolor}

\definecolor{oursbg}{RGB}{232,245,238}
\definecolor{sectionbg}{RGB}{242,245,248}
\definecolor{costbg}{RGB}{250,243,230}

\usepackage{algorithm}
\usepackage{algpseudocode}
\usepackage{float}
\usepackage{placeins}
\usepackage{needspace}

\providecolor{MAMGBlue}{HTML}{0072B2}
\providecolor{MAMGOrange}{HTML}{D55E00}
\providecolor{MAMGGreen}{HTML}{009E73}
\providecolor{MAMGPurple}{HTML}{AA4499}
\providecolor{MAMGInk}{HTML}{183B4A}
\providecolor{MAMGGray}{HTML}{5F6B70}

\newsavebox{\MoMPromptSave}
\newlength{\MoMPromptWidth}
\newenvironment{mompromptbox}[2][]{%
  \def\MoMPromptTitle{#1}%
  \def\MoMPromptColor{#2}%
  \par\medskip\noindent
  \begin{lrbox}{\MoMPromptSave}%
  \begin{minipage}{\MoMPromptWidth}%
  \colorbox{#2!10!white}{%
    \parbox{\dimexpr\MoMPromptWidth-2\fboxsep\relax}{%
      \color{MAMGInk}\normalfont\sffamily\small\bfseries #1}}%
  \par\smallskip
  \begingroup\ttfamily\footnotesize\raggedright
}{%
  \par\endgroup
  \end{minipage}%
  \end{lrbox}%
  \begingroup
  \setlength{\fboxrule}{0.55pt}%
  \fcolorbox{\MoMPromptColor!72!black}{\MoMPromptColor!3!white}{\usebox{\MoMPromptSave}}%
  \endgroup
  \par\medskip
}

\newcommand{\promptslot}[1]{{\color{MAMGPurple}\texttt{\{#1\}}}}
\newcommand{\promptcontract}[1]{{\color{MAMGOrange}#1}}
\newcommand{\promptcomment}[1]{{\color{MAMGGray}\normalfont\itshape #1}}
\newcommand{\promptbadge}[2]{%
  \begingroup
  \setlength{\fboxrule}{0.45pt}%
  \fcolorbox{#1!65!black}{#1!8!white}{%
    \color{MAMGInk}\sffamily\scriptsize\bfseries\strut #2}%
  \endgroup
}

\newtheorem{proposition}{Proposition}
\definecolor{MoMFailure}{HTML}{D97706}
\definecolor{MoMRepair}{HTML}{0072B2}
\definecolor{MoMRequest}{HTML}{E67E00}
\definecolor{MoMAssist}{HTML}{009DB2}
\definecolor{MoMHarm}{HTML}{C43C35}
\definecolor{MoMValue}{HTML}{00876C}
\definecolor{MoMMarket}{HTML}{00875A}

\newcommand{\momrequest}[1]{\textcolor{MoMRequest}{#1}}
\newcommand{\momassist}[1]{\textcolor{MoMAssist}{#1}}
\newcommand{\momharm}[1]{\textcolor{MoMHarm}{#1}}
\newcommand{\momvalue}[1]{\textcolor{MoMValue}{#1}}
\newcommand{\mommarket}[1]{\textcolor{MoMMarket}{#1}}

\usepackage{url}
\usepackage{hyperref}

\title{Markets of Minds: A 51-Billion-State Economy of Intelligence}

\author{%
  \textbf{Jusheng Zhang$^{1}$, Xuekun Wu$^{1}$, Wenhao Wang$^{2,\dagger}$, Mark Wang$^{2,\dagger}$, Ming-Hsuan Yang$^{3}$}\\
  $^{1}$Stanford University,
  $^{2}$Vast Intelligence Lab,
  $^{3}$University of California at Merced
}

\iclrfinalcopy

\begin{document}

\maketitle
% The ICLR style sets the conference header inside \maketitle.
% Clear it afterwards for the arXiv copy.
\fancyhead{}
\renewcommand{\headrulewidth}{0pt}

% Shared dagger footnote for the two corresponding authors.
\begingroup
\renewcommand{\thefootnote}{\fnsymbol{footnote}}
\footnotetext[2]{Co-corresponding authors:
  Wenhao Wang and
  Mark Wang.}
\endgroup
\vspace{-4mm}
\begin{abstract}
\vspace{-3mm}
When intelligence is shared, reasoning becomes a constrained allocation problem.
In a population of agents with limited access to expert models, an intervention can
benefit its recipient yet still be a poor use of shared capacity if another feasible
intervention has greater expected value. Assistance must therefore be evaluated
through both its marginal benefit and the opportunity cost induced by budgets,
access, scarcity, and service capacity. To this end, we introduce Markets of Minds,
a framework for allocating expert reasoning across agents before expert outputs
are observed. For each feasible agent--service pair, the framework
estimates the expected gain from assistance relative to the agent's default
continuation, explicitly separating repair benefit from regression risk, and jointly
clears assistance contracts under shared resource constraints. We then study this
problem in a dependency-aware agent economy with 51.2 billion addressable
persona--market states spanning capabilities, cognitive budgets, provider access,
scarcity, and network roles. These dependencies couple allocation feasibility
across agents rather than the outcome of a fixed contract. Empirically, in a matched
replay of 240 cases under a common 30\% assistance quota, the frozen allocator
reaches 85.42\% task-level persona adherence versus 73.86\% for equal-capacity
random allocation, a gain of 11.56 percentage points. The results show that
assistance value is heterogeneous and can be negative, and that scarcity
changes which interventions are worth purchasing. Together, these findings
motivate a population-level view of intelligence where the value of assistance
depends not only on what it can achieve, but also on which alternative uses of
scarce cognitive capacity it displaces.

\end{abstract}

\vspace{-4mm}
\section{Introduction}
\label{sec:introduction}
Intelligence can be abundant in capability yet scarce in access. Large language
models increasingly serve as agents that browse, plan, use tools, and act on
behalf of different users \citep{wu2023autogen,hong2023metagpt}. Expanding this population, however, does not
necessarily expand the cognitive services available to it. Access to stronger
reasoning, specialized expertise, long-context processing, or reliable
verification may remain bounded by cost, latency, concurrency, and provider
restrictions \citep{chen2023frugalgpt,ong2024routellm,feng2024graphrouter}. A system can therefore support many capable agents while lacking
the capacity to assist all of them when needed. This creates a population-level
question: \emph{when many agents need help, which requests should receive scarce
cognition?}

Need alone is not value under scarcity. Under shared capacity, assistance
decisions are coupled: assigning a service slot to one agent may displace an
intervention with greater expected benefit elsewhere, while spending resources
on one request reduces what remains available to others
\citep{mehta2007adwords,badanidiyuru2013bandits,zhang2022reusable}. The gain from help is
also uncertain and can be negative \citep{aggarwal2023automix,ong2024routellm}. An
agent may already have a sufficient answer, an expert may fail to repair its
difficulty, or additional advice may disrupt an otherwise correct solution. The
allocator must therefore distinguish the expected gain from assistance relative
to the agent's default continuation from the opportunity cost of committing
shared resources to it. A locally beneficial intervention can still be a poor
population-level allocation \citep{badanidiyuru2013bandits,zhang2022reusable}.
Crucially, this decision must be made before observing the corresponding expert
output.

We formalize this setting as \emph{multi-agent assistance under shared cognitive
budgets} (Section~\ref{sec:problem}). In a decision window, each request either
continues under its default policy or receives at most one feasible typed
service. The objective is to maximize expected downstream task performance
across the population subject to individual budgets, provider access, shared
resource limits, service capacity, and deadlines. These constraints determine
which combinations of interventions are feasible, so resources
committed to one agent restrict the choices available to others. Dependencies
across agents enter through budgets, access, scarcity, capacity, network roles,
and the feasible set; conditional on a fixed contract, we do not
model another agent's assignment as changing its task outcome.

We introduce \emph{Markets of Minds} (Section~\ref{sec:method}), a framework
that represents cognitive assistance as bounded, typed contracts and allocates
them before expert execution. For each feasible agent--service pair, the
framework estimates the marginal value of assistance relative to the agent's
default continuation. The formulation separates three quantities: the
probability that the default continuation already succeeds, the probability
that a service repairs a default failure, and the risk that assistance regresses
an otherwise successful outcome. A joint clearing step converts these values into contracts subject to budget, access, capacity, and
deadline constraints, accounting for competing uses of resources \citep{zhang2026uncertainty,zhou2026agora,alazraki2026sale}.
Only selected services are executed, and a common ledger records reservation,
consumption, and settlement. The problem is therefore not to select among
already generated responses, but to decide which responses are worth acquiring
before their realized value is known.

To study how economic context shapes this decision, we construct a
dependency-aware agent economy that crosses 999,847 base personas with
10 cognitive budgets, 20 expertise portfolios, 8 provider-access regimes,
4 scarcity regimes, and 8 network roles. The resulting design space contains
exactly \(51{,}192{,}166{,}400\) addressable persona--market states. 
A versioned sampling procedure draws auditable cohorts, and experiments
materialize only the states required for a given hypothesis. The construction
allows budgets, expertise, provider access, supply conditions, and network
roles to vary over the same underlying population. Economic context therefore
becomes an experimental variable rather than a fixed background assumption,
while persona demographics remain distinct from the capabilities and resources
assigned to an agent.

Our evaluation separates evidence about allocation from evidence about service
capability and population scale. The primary matched replay deliberately
instantiates a one-service setting with a fixed assistance quota, isolating the
question of \emph{who should receive help}. Across 240 interactive agent cases
with actor and judge models held fixed, all deployable policies are evaluated
against the same cached default and assisted outcomes under a common 30\%
assistance quota. A frozen allocator achieves 85.42\% task-level persona
adherence, compared with 73.86\% for equal-capacity random allocation, a gain
of 11.56 percentage points without increasing the number of assisted cases.
A complementary study across six multimodal benchmarks characterizes where
different services are useful and where their failures concentrate. These
experiments address distinct questions: recipient selection under fixed supply
and heterogeneity in service capability. The broader heterogeneous-service
clearing formulation defines the general framework, while empirical claims are
tied to the evaluated instantiations.
Together, the formulation, framework, and experiments support a population-level view of cognitive assistance under shared constraints. The key question is not simply whether expert help improves an agent's outcome, but whether it is the best use of scarce capacity relative to the alternatives it displaces. Intelligence is thus both an individual capability and a shared resource, whose value depends on both intervention gain and opportunity cost.
\vspace{-2mm}
\section{Problem Formulation}
\label{sec:problem}

We study how a population of task agents should share limited access to
cognitive services such as specialist reasoning, verification, and tool
support. Consider $n$ agents and $m$ services in a single decision window.
Agent $i$ has an \momrequest{observable state} $\momrequest{s_i}$, comprising
its task, execution history, and available resources. Let
$\momrequest{Y_i^0}\in\{0,1\}$ denote \momrequest{default success} and
$\momassist{Y_{ij}^1}\in\{0,1\}$ denote \momassist{assisted success} after
receiving service $j$. We assume that $\momrequest{s_i}$ is sufficient for
predicting these outcomes and that other agents' assignments affect feasibility
but do not change the outcome of a fixed contract. Thus cross-agent dependence
enters through shared budgets, access, scarcity, capacity, and the feasible
allocation set.

Allocation maximizes the expected \momvalue{assistance gain}:
\begin{equation}
    \momvalue{\Delta_{ij}}
    \triangleq \mathbb{E}[\momassist{Y_{ij}^1}-\momrequest{Y_i^0}
    \mid\momrequest{s_i}],
    \qquad
    \mommarket{x^\star}\in
    \arg\max_{\mommarket{x}\in\mommarket{\mathcal{X}}}
    \sum_{i,j}\mommarket{x_{ij}}\momvalue{\Delta_{ij}},
    \label{eq:assistance-delta}
\end{equation}
where the \mommarket{assignment} $\mommarket{x_{ij}}$ indicates whether agent
$i$ receives service $j$, and the \mommarket{feasible set}
$\mommarket{\mathcal{X}}$ encodes capacity, budget, and access constraints with
at most one service per request. Allocation must precede observing the assisted
response and either outcome label. Section~\ref{sec:method} instantiates this
objective using observable request features
(Equation~\eqref{eq:demand-representation}) and estimated contract values
(Equation~\eqref{eq:estimated-value}).

\noindent\textbf{Assistance value separates need from service effect.}
A low probability of default success indicates that an agent may need help,
but does not establish that an available service can repair its failure.
Assistance can also disrupt an otherwise successful continuation. For binary
outcomes, these two effects admit an exact decomposition.

\begin{proposition}[Assistance value]
\label{prop:value-decomposition}
The expected value of an assistance contract factorizes into default-state
occupancy and the signed transition induced by the service:
\vspace{-2mm}
\begin{equation}
    \vspace{-2mm}
    \momvalue{\Delta_{ij}}
    =
    \begin{bmatrix}\momrequest{1-p_i}&\momrequest{p_i}\end{bmatrix}
    \begin{bmatrix}\momassist{r_{ij}}\\-\momharm{h_{ij}}\end{bmatrix}
    =
    \momrequest{(1-p_i)}\momassist{r_{ij}}
    -\momrequest{p_i}\momharm{h_{ij}},
    \label{eq:value-decomposition}
\end{equation}
    \vspace{-3mm}
\end{proposition}

Here, $\momrequest{p_i}=\Pr(\momrequest{Y_i^0}=1\mid\momrequest{s_i})$ is
\momrequest{default success}, $\momassist{r_{ij}}=\Pr(\momassist{Y_{ij}^1}=1\mid
\momrequest{Y_i^0}=0,\momrequest{s_i})$ is conditional
\momassist{repairability}, and
$\momharm{h_{ij}}=\Pr(\momassist{Y_{ij}^1}=0\mid
\momrequest{Y_i^0}=1,\momrequest{s_i})$ is conditional
\momharm{regression risk}. The identity follows by conditioning on
$\momrequest{Y_i^0}$:
$\mathbb{E}[\momassist{Y_{ij}^1}\mid\momrequest{s_i}]
=\momrequest{(1-p_i)}\momassist{r_{ij}}
+\momrequest{p_i}(1-\momharm{h_{ij}})$.
The $1\times2$ row gives \momrequest{default-state occupancy}; the $2\times1$
column gives the signed \momassist{repair} and \momharm{regression} effects.
The key distinction is therefore between \emph{needing help} and
\emph{benefiting from the available help}. A complete proof appears in
Appendix~\ref{app:problem-proofs}.

\begin{figure}[!t]
    \centering
    \includegraphics[width=1\linewidth]{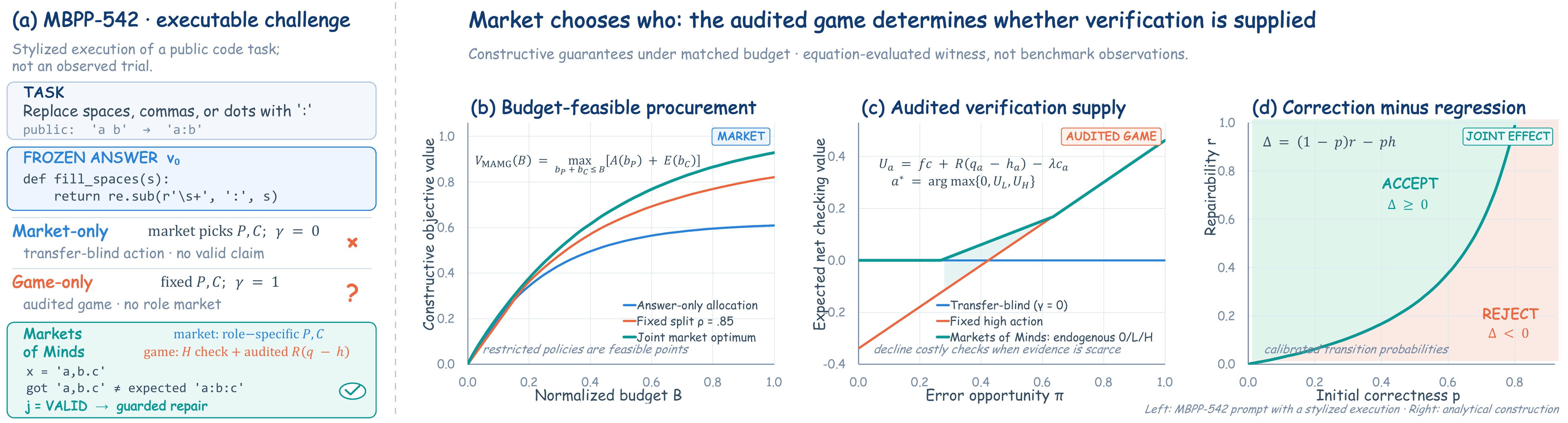}
    \vspace{-6mm}
    \caption{\textbf{Allocation and verification.}
    A stylized execution and analytical illustrations of procurement,
    verification supply, and net correction.}
    \label{fig:allocation-verification}
    \vspace{-4mm}
\end{figure}

Fig.~\ref{fig:allocation-verification} combines executable checking~(a),
budget-feasible procurement~(b), verification supply~(c), and net
correction~(d) in Equation~\eqref{eq:value-decomposition}. These constructions
are analytical illustrations rather than benchmark observations, and our
allocation model does not assume the depicted audited game. Panel~(d)'s accept
region means non-negative gain, not guaranteed selection. At fixed regression
risk $h$ and $p<1$, its boundary is $r=ph/(1-p)$.

\noindent\textbf{Difficulty alone can misallocate assistance.}
Consider requests with $(p_i,r_{ij},h_{ij})=(0.1,0,0)$ and $(0.6,1,0)$. The
first is more likely to fail, yet has assistance value zero; the second has
value $0.4$. With one slot, a difficulty-first rule selects the irreparable
case. Conversely, even a capable service can have negative value when
\momharm{regression} outweighs \momassist{repair}. Section~\ref{sec:method}
therefore estimates service effects before forming
Equation~\eqref{eq:estimated-value}, rather than treating difficulty itself as
assistance value.

\noindent\textbf{Population welfare depends on value estimation.}
Expected population reward is
$W(\mommarket{x})=\sum_i\momrequest{p_i}
+\sum_{i,j}\mommarket{x_{ij}}\momvalue{\Delta_{ij}}$: a fixed default baseline
plus the reward added by allocation. For one service and exactly $K$
assignments, the optimum selects the $K$ largest $\momvalue{\Delta_i}$ values;
when $K$ is only an upper bound, non-positive values need not be selected.
Because these contract values are unknown before clearing, errors in value
estimation can translate directly into errors in population allocation.

\begin{proposition}[Allocation error]
\label{prop:allocation-error}
Let $\mommarket{\widehat{x}}$ maximize
$\sum_{ij}\mommarket{x_{ij}}\momvalue{\widehat{\Delta}_{ij}}$ over the
feasible set of $\mommarket{x^\star}$. If allocations contain at most $K$
contracts and
$|\momvalue{\widehat{\Delta}_{ij}}-\momvalue{\Delta_{ij}}|\leq\momharm{\epsilon}$
for all feasible pairs, then
\begin{equation}
    W(\mommarket{x^\star})-W(\mommarket{\widehat{x}})
    \leq
    \sum_{i,j}(\mommarket{x^\star_{ij}}+\mommarket{\widehat{x}_{ij}})
    |\momvalue{\widehat{\Delta}_{ij}}-\momvalue{\Delta_{ij}}|
    \leq 2K\momharm{\epsilon}.
    \label{eq:regret-bound}
\end{equation}
    \vspace{-6mm}
\end{proposition}

The first inequality follows by inserting the predicted objectives and using
optimality of $\mommarket{\widehat{x}}$. Each allocation contains at most $K$
contracts; the \momharm{uniform estimation error} $\momharm{\epsilon}$ gives
the second inequality. For an exact top-$K$ quota with $1\leq K<n$, the selected
set also matches the oracle whenever
$\Delta_{(K)}-\Delta_{(K+1)}>2\epsilon$, with values ordered from largest to
smallest. This is a conditional guarantee: average predictive calibration alone
does not establish the required uniform-error premise. The complete argument,
including the top-$K$ recovery condition, appears in
Appendix~\ref{app:problem-proofs}.

Together, Propositions~\ref{prop:value-decomposition} and
\ref{prop:allocation-error} connect individual assistance effects to
population-level allocation. The first identifies what makes an intervention
valuable; the second shows why estimating that value accurately matters once
multiple agents compete for the same cognitive capacity.

\FloatBarrier

\vspace{-2mm}
\section{Markets of Minds}
\label{sec:method}

The formulation above separates two requirements for allocating scarce
cognition. Proposition~\ref{prop:value-decomposition} identifies what must be
estimated for each feasible contract, \textit{i.e.}, \momrequest{default success},
conditional \momassist{repair}, and \momharm{regression risk}, while
Proposition~\ref{prop:allocation-error} explains why errors in these values
matter once agents compete for shared capacity. \emph{Markets of Minds}
turns this decomposition into a pre-execution allocation pipeline: represent
observable cognitive demand, estimate prospective assistance value, clear
contracts jointly under shared constraints, and execute only the selected
assistance. Fig.~\ref{fig:mom-overview} summarizes this process.

\begin{figure}[!t]
    \centering
    \includegraphics[width=\linewidth]{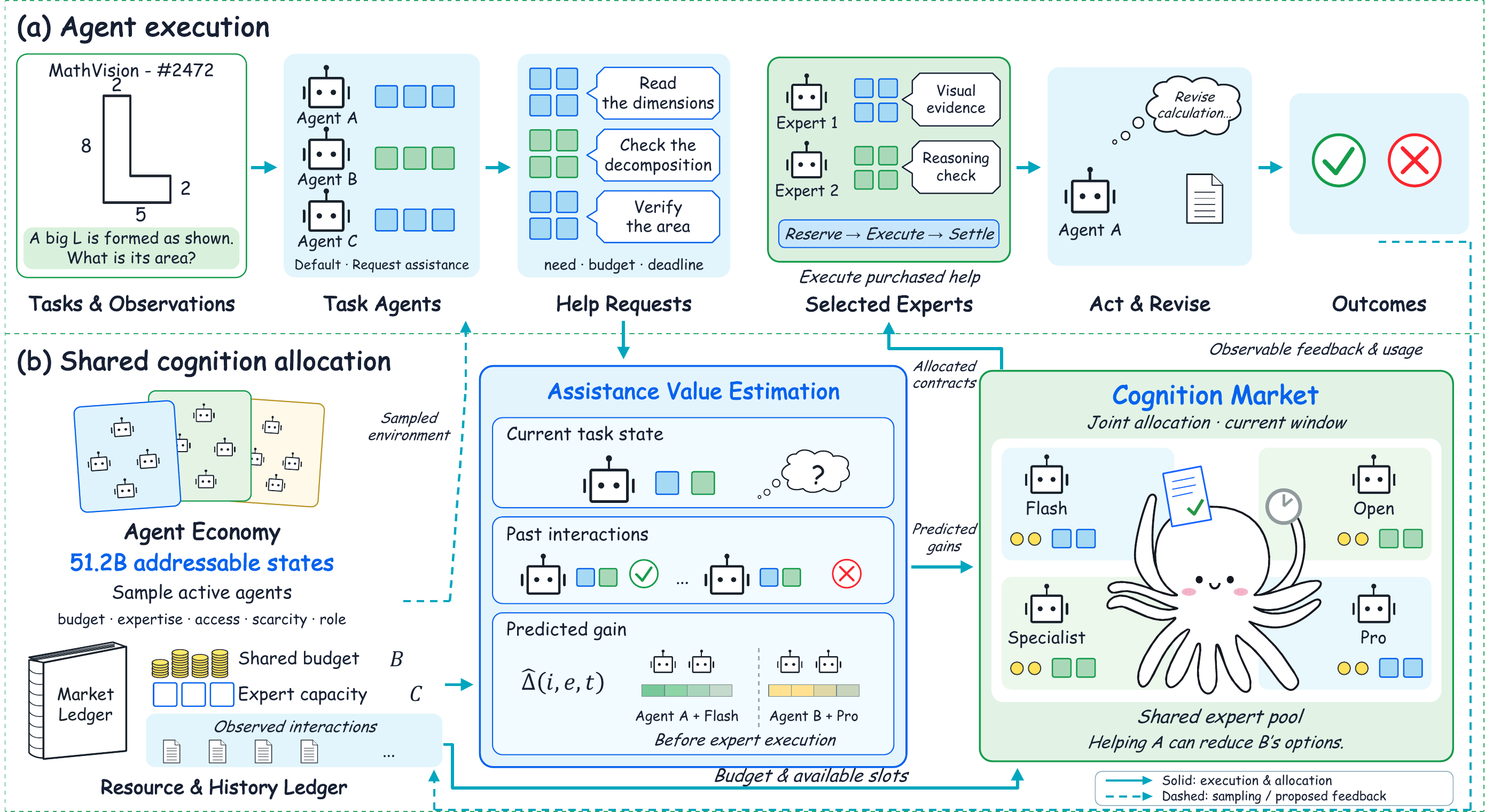}
    \vspace{-6mm}\caption{\textbf{Markets of Minds.} (a) Task-agent execution. (b) Pre-execution
    value estimation, joint allocation, and resource accounting.}
    \label{fig:mom-overview}
    \vspace{-4mm}
\end{figure}

The \momvalue{value-estimation} block implements
Equations~\eqref{eq:demand-representation} and~\eqref{eq:estimated-value}; the
\mommarket{market} clears contracts under the feasible set of
Section~\ref{sec:problem}. In the diagram, $\momvalue{\widehat\Delta(i,e,t)}$
denotes $\momvalue{\widehat\Delta_{ij}}$ with expert $e=j$ and the window $t$
fixed. The ledger supplies the budgets and capacities available at clearing
time. Reservation, execution, and settlement occur only after the allocation
has been fixed, so unobserved expert outputs cannot influence the decision that
selects them.

\noindent\textbf{\momrequest{Represent observable cognitive demand.}}
Agent $i$ submits a \momrequest{demand packet}
$\momrequest{D_i}=(q_i,o_i,\theta_i,g_i,b_i,\mathcal{A}_i,d_i)$ containing its
task, observable base artifact, task-relevant persona context, desired behavior
or missing evidence, remaining credits, accessible services, and deadline.
The packet excludes both the target outcome and any unexecuted service response,
and therefore provides a structured view of the \momrequest{observable state}
$\momrequest{s_i}$ introduced in Section~\ref{sec:problem}. A demand auditor
maps it to
\begin{equation}
    \momrequest{\mathbf{a}_i}
    = (u_i^{\mathrm{per}},u_i^{\mathrm{sem}},u_i^{\mathrm{inf}},\kappa_i,m_i)^\top,
    \qquad
    \mommarket{\mathbf{z}_i}
    = [\momrequest{\mathbf{a}_i};\,\mommarket{\boldsymbol{\rho}_i}],
    \qquad
    \momassist{\mathbf{z}_{ij}}
    = [\mommarket{\mathbf{z}_i};\,\momassist{\mathbf{e}_j}].
    \label{eq:demand-representation}
\end{equation}
The five \momrequest{demand features} in $\momrequest{\mathbf a_i}$ identify
possible failures in perception, interpretation, inference, preference
resolution, and execution; they are observable features rather than outcome
labels. The \mommarket{request representation} $\mommarket{\mathbf z_i}$ adds
\mommarket{historical statistics} $\mommarket{\boldsymbol\rho_i}$ from similar
development requests. The \momassist{request--service representation}
$\momassist{\mathbf z_{ij}}$ further adds the \momassist{service descriptor}
$\momassist{\mathbf e_j}$ so that the predicted effect can depend on which
service is considered. Here $[\cdot;\cdot]$ denotes column concatenation; the
three vectors have dimensions $5$, $d_{\mathrm{req}}$, and $d$, respectively.
With one service, $\momassist{\mathbf e_j}$ is constant across requests and
does not affect their ranking.

\noindent\textbf{\momvalue{Estimate prospective assistance value.}}
The relevant quantity is not whether a request appears difficult, but whether
a particular service is expected to improve its outcome relative to the
default continuation. On paired development cases, we can observe both the
default and assisted continuations: cases with $\momrequest{Y_i^0}=0$ reveal
\momassist{repair}, and cases with $\momrequest{Y_i^0}=1$ reveal
\momharm{regression}. The full component-wise estimator fits
$\momrequest{\widehat p_i}$ on all development cases using request features
$\mommarket{\mathbf z_i}$, then fits $\momassist{\widehat r_{ij}}$ on default
failures and $\momharm{\widehat h_{ij}}$ on default successes using
request--service features $\momassist{\mathbf z_{ij}}$. Each head minimizes
binary cross-entropy on its own outcome stratum; regularization and group
smoothing are selected only on development folds.

The frozen estimates instantiate the decomposition in
Equation~\eqref{eq:value-decomposition} and the objective in
Equation~\eqref{eq:assistance-delta} through one plug-in decision:
\begin{equation}
\vspace{-3mm}
    \momvalue{\widehat{\Delta}_{ij}}
    = \momrequest{(1-\widehat p_i)}\momassist{\widehat r_{ij}}
      -\momrequest{\widehat p_i}\momharm{\widehat h_{ij}},
    \qquad
    \mommarket{\widehat x}\in
    \arg\max_{\mommarket{x}\in\mommarket{\mathcal X}}
    \sum_{i,j}\mommarket{x_{ij}}\momvalue{\widehat\Delta_{ij}}.
    \label{eq:estimated-value}
\end{equation}
The \momrequest{default-success estimate} $\momrequest{\widehat p_i}$ weights
the two initial states; \momassist{repairability}
$\momassist{\widehat r_{ij}}$ and \momharm{regression risk}
$\momharm{\widehat h_{ij}}$ give their conditional transition effects. Their
weighted difference is the \momvalue{estimated contract gain}
$\momvalue{\widehat\Delta_{ij}}$.
\textbf{\mommarket{Clear contracts under shared constraints.}}
The second part of Equation~\eqref{eq:estimated-value} converts these local
contract values into a population-level allocation. The
\mommarket{feasible set} $\mommarket{\mathcal X}$ imposes one service per
request, buyer budgets, supplier capacities, shared resource limits including
broker overhead, access permissions, and deadlines. Consequently, a contract
can have positive estimated value and still remain unselected when another
feasible use of the same capacity has greater value. This is precisely the
opportunity-cost effect represented by the population objective in
Section~\ref{sec:problem}. Proposition~\ref{prop:allocation-error} connects
errors in these estimated values to the resulting allocation loss.
Algorithm~\ref{alg:mom} separates development fitting from target inference and executes only cleared contracts.

\noindent\textbf{Evaluated instantiation.}
The matched replay deliberately isolates the recipient-selection
problem. It has one service and exactly $K_e$ slots per environment. The
auditor also supplies default-success and repair-gain estimates
$\widetilde p_i$ and $\widetilde g_i$. The implemented calibrator fits
$v_i=f_\beta(\mathbf a_i,\widetilde p_i,\widetilde g_i,\boldsymbol\rho_i)
\approx\Pr(Y_i^0=0,Y_i^1=1\mid D_i)$ on development cases; clearing then selects
the top $K_e$ requests by $v_i$. This estimates the positive-effect term of
Equation~\eqref{eq:value-decomposition}; it does not separately estimate
regression risk. The current result therefore tests \emph{demand-aware
recipient selection under fixed assistance supply}. The three-head estimator and heterogeneous-service clearing in
Algorithm~\ref{alg:mom} define the general framework, while the experiments
instantiate it in the one-service recipient-selection setting.

\noindent\textbf{\mommarket{Execute and account only after clearing.}}
Only after the allocation is frozen do selected agents receive the service;
unselected agents continue by default. Matched replay uses the same
trajectories for every policy and joins target outcomes after selection. This
separation ensures that the allocator cannot use the realized assisted outcome
to decide whether assistance should have been purchased. The ledger records
outcomes, credits, resource use, latency, failures, and API-equivalent cost.
In Algorithm~\ref{alg:mom}, $\widehat{\mathcal C}$ is the purchased contract
set, $\mathcal O$ contains observed execution results and usage, and
$\mathcal L^-,\mathcal L^+$ are the reserved and settled ledger states. The
dashed feedback path in Fig.~\ref{fig:mom-overview} indicates possible
future reuse of logged interactions; it is not updated on target-window outcomes.

\vspace{-2mm}
\section{Experiments}
\label{sec:experiments}

We evaluate the empirical questions induced by
Sections~\ref{sec:problem}--\ref{sec:method}. We first test whether a frozen
pre-execution allocator improves recipient selection when every policy receives
the same expert capacity. We then examine whether the resulting gains persist
across heterogeneous multimodal tasks and expert-use budgets, followed by an
explicit shared-capacity population setting. Finally, we ablate the signals used
by the implemented allocator. These experiments evaluate the instantiated allocation policies in
Section~\ref{sec:method}, with the matched replay focusing on one-service
recipient selection and the remaining studies examining task heterogeneity
and shared-capacity allocation.

\subsection{Matched Agent Allocation}
\label{sec:matched-agent-allocation}

\textbf{Protocol.}
We first evaluate the one-service recipient-selection setting described in
Section~\ref{sec:method}. The matched replay contains 240 interactive agent
cases, with 80 cases each from Chat, Web, and OS-App. Every case has cached
default and assisted potential outcomes that are shared across allocation
policies. GPT-5.6 Sol serves as the actor and Claude Opus 4.8 as the judge.
Only the Survey split is used for development. Within each environment, an
assistance policy may allocate at most 24 expert interventions, yielding 72
interventions in total and a common 30\% expert-capacity limit. Thus all
policies operate on the same cases, expert supply, and evaluation outcomes.
Overall denotes task-level persona adherence.

\begin{table}[t]
\caption{\textbf{Matched allocation at 30\% shared expert capacity.}
80 cases and 24 assistance slots per environment; cost/success includes
method-specific overheads.}
\vspace{1mm}
\label{tab:matched_agent_allocation}
\centering
\small
\setlength{\tabcolsep}{1.8pt}
\renewcommand{\arraystretch}{1.05}
\begin{tabular*}{\linewidth}{@{\extracolsep{\fill}}lrrrrrrr@{}}
\toprule
Method
& Chat
& Web
& OS-App
& \shortstack{Overall $\uparrow$}
& \shortstack{Expert use}
& \shortstack{Cost/succ. $\downarrow$}
& \shortstack{Harmful help $\downarrow$} \\
\midrule
No assistance
& 68.75 & 61.25 & 65.00 & 65.00
& 0\% & \$0.06211 & -- \\

Random
& 77.40 & 70.59 & 73.60 & 73.86
& 30\% & \$0.06090 & 0.42\% \\

Difficulty/uncertainty
& 83.75 & 76.25 & 85.00 & 81.67
& 30\% & \$0.05636 & 0.00\% \\

AutoMix
& 82.50 & 76.25 & 81.25 & 80.00
& 30\% & \$0.07765 & 0.00\% \\

RouteLLM
& 86.25 & 77.50 & 80.00 & 81.25
& 30\% & \$0.06219 & 0.00\% \\

GraphRouter
& 86.25 & 77.50 & 80.00 & 81.25
& 30\% & \$0.06023 & 0.00\% \\

Agora
& 86.25 & \textbf{80.00} & 82.50 & 82.92
& 30\% & \$0.05778 & \textbf{0.00\%} \\

\midrule
\textbf{Ours}
& \textbf{91.25}
& 78.75
& \textbf{86.25}
& \textbf{85.42}
& \textbf{30\%}
& \textbf{\$0.05509}
& \textbf{0.00\%} \\
\bottomrule
\end{tabular*}
\vspace{-4mm}
\end{table}

\noindent\textbf{Observation 1: Recipient selection matters even when
assistance supply is fixed.}
Table~\ref{tab:matched_agent_allocation} shows that Markets of Minds reaches
\textbf{85.42} overall under the same 30\% expert-capacity limit used by all
assistance policies. Equal-capacity random allocation reaches 73.86, giving an
\textbf{11.56-point} difference, while the strongest comparison policy, Agora,
reaches 82.92, leaving a \textbf{2.50-point} gap. Markets of Minds is highest
on Chat and OS-App, whereas Agora is higher on Web, indicating that the value
of recipient selection varies across environments. Because the cases, cached
potential outcomes, and expert supply are matched across policies, this
comparison isolates allocation quality rather than additional expert use.
The same table also shows that higher task performance does not require higher
realized expert consumption: all assistance policies use 30\% capacity.
Markets of Minds obtains the lowest cost per successful task
(\textbf{\$0.05509}) and records no harmful-help events in this replay, tied
with several comparison policies on the latter diagnostic.
\begin{figure*}[t]
\centering
\includegraphics[width=1\textwidth]{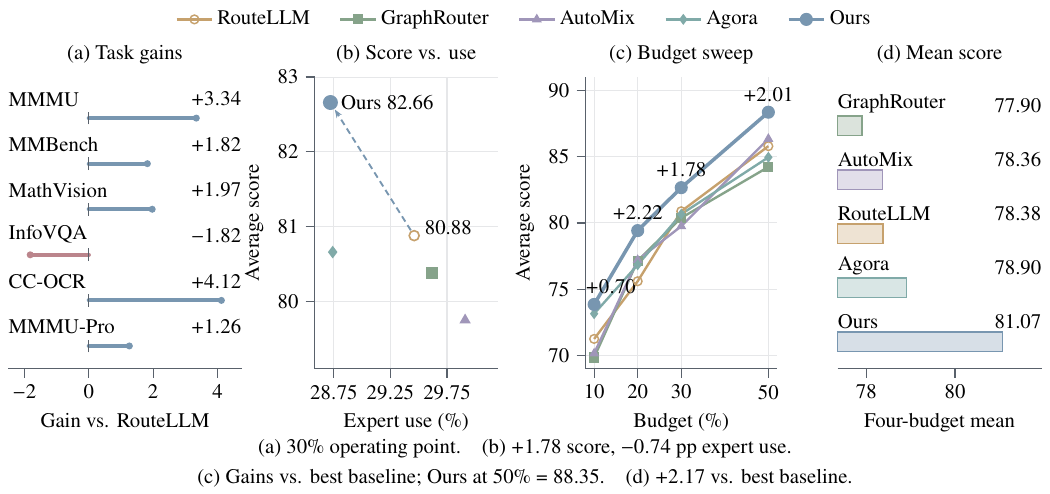}
\vspace{-6mm}
\caption{\textbf{Main routing results and budget robustness.}
Task-level gains, score versus expert use, performance across four expert-use
budgets, and the corresponding four-budget mean.}
\label{fig:experiments-overview}
\vspace{-4mm}
\end{figure*}

\subsection{Multimodal Benchmark Results}
\label{sec:main-benchmark-results}

\begin{table}[t]
\caption{\textbf{Multimodal benchmark results.}
Routing entries report three-run means; \textbf{bold} marks the best routing
result in each column.}
\vspace{1mm}
\label{tab:main-results}
\centering
\small
\setlength{\tabcolsep}{1.8pt}
\renewcommand{\arraystretch}{1.05}
\begin{tabular*}{\linewidth}{@{\extracolsep{\fill}}lrrrrrrrr@{}}
\toprule
Method &
MMMU &
MMBench &
\shortstack{Math Vision} &
InfoVQA &
CC-OCR &
MMMU-Pro &
Avg. &
\shortstack{Expert use} \\
\midrule
\multicolumn{9}{l}{\textit{Single-model references}} \\
Gemini 3.5 Flash
& 76.11 & 92.99 & 86.51 & 78.90 & 71.42 & 78.03 & 80.66 & --- \\
Gemini-2.5-Pro
& 81.70 & 88.30 & 63.50 & 81.00 & 73.00 & 68.00 & 75.92 & --- \\
Gemini-2.0-Flash
& 70.70 & 83.00 & 41.30 & 83.20 & 73.10 & 57.00 & 68.05 & --- \\
GPT-5.6 Sol
& 82.22 & 83.90 & 84.54 & 73.58 & 79.62 & 79.77 & 80.60 & --- \\
GPT-4o (2024-08-06)
& 70.70 & 74.30 & 30.40 & 68.70 & 66.60 & 54.50 & 60.87 & --- \\
Claude Sonnet 5
& 65.00 & 83.90 & 74.67 & 62.31 & 74.26 & 69.94 & 71.68 & --- \\
InternVL3-78B
& 72.20 & 87.70 & 43.10 & 84.10 & 80.30 & 55.10 & 70.42 & --- \\
Qwen2.5-VL-72B
& 70.20 & 88.40 & 38.10 & 87.30 & 79.80 & 51.10 & 69.15 & --- \\
\midrule
\multicolumn{9}{l}{\textit{Routing methods (three-run means)}} \\
RouteLLM
& 74.44 & 91.69 & 84.38 & \textbf{78.84}
& 76.73 & 79.19 & 80.88 & 29.46\% \\
GraphRouter
& 74.12 & 91.45 & 84.05 & 78.15
& 76.10 & 78.43 & 80.38 & 29.62\% \\
AutoMix
& 73.55 & 90.72 & 83.81 & 77.64
& 74.98 & 77.82 & 79.75 & 29.91\% \\
Agora
& \textbf{78.33} & 92.73 & 86.02 & 76.87
& 73.14 & 76.88 & 80.66 & 28.74\% \\
\textbf{Ours}
& 77.78 & \textbf{93.51} & \textbf{86.35} & 77.02
& \textbf{80.85} & \textbf{80.45} & \textbf{82.66}
& \textbf{28.72\%} \\
\bottomrule
\end{tabular*}
\vspace{-4mm}
\end{table}

\textbf{Benchmarks and metrics.}
We complement the matched agent replay with a six-benchmark multimodal study
using MMMU~\citep{yue2024mmmu},
MMBench~\citep{liu2024mmbench},
MATH-Vision~\citep{wang2024mathvision},
InfoVQA~\citep{mathew2022infographicvqa},
CC-OCR~\citep{yang2025ccocr}, and
MMMU-Pro~\citep{yue2025mmmupro}.
We evaluate on the full designated evaluation splits and report
benchmark-specific scores on a 0--100 scale together with their unweighted
mean (\textit{Avg.}). Table~\ref{tab:main-results} reports three-run means and
realized expert use, while Fig.~\ref{fig:experiments-overview} summarizes
task-level differences and performance across expert-use budgets.

\textbf{Models and comparisons.}
Gemini 3.5 Flash~\citep{google2026gemini35flash} provides the base response.
The assistance pool comprises Gemini-2.5-Pro~\citep{google2024gemini25},
GPT-5.6 Sol~\citep{openai2026gpt56sol}, Claude Sonnet
5~\citep{anthropic2026sonnet5}, InternVL3-78B~\citep{InternVL3}, and
Qwen2.5-VL-72B. We compare Markets of Minds with matched adaptations of
RouteLLM~\citep{ong2024routellm},
GraphRouter~\citep{feng2024graphrouter},
AutoMix~\citep{aggarwal2023automix}, and
Agora~\citep{zhang2026agora} under a common assistance interface.
Standalone models are reported separately as performance references rather
than resource-matched allocation baselines.

\textbf{Allocation and reporting.}
Each request may receive at most one expert intervention; unselected requests
retain their base response. All routing policies share the same response pool,
eligible requests, and resource constraints, with a nominal assistance rate of
30\%. Allocation decisions are fixed before the corresponding evaluation
outcomes are used for scoring. Expert use is the realized percentage of
evaluated requests assigned to an expert, averaged over three runs.

\noindent\textbf{Observation 2: Assistance value is heterogeneous across
tasks.}
At the 30\% operating point, Table~\ref{tab:main-results} shows that Markets of
Minds reaches \textbf{82.66} average score with \textbf{28.72\%} expert use,
compared with 80.88 and 29.46\% for RouteLLM, the highest-scoring comparison
router by average score. The difference is therefore
\textbf{+1.78 points} while using \textbf{0.74 percentage points fewer}
expert calls. The gain is not uniform across benchmarks: Markets of Minds has
the highest routing score on MMBench, MATH-Vision, CC-OCR, and MMMU-Pro,
whereas Agora is higher on MMMU and RouteLLM is higher on InfoVQA. This
variation is consistent with the formulation in Section~\ref{sec:problem}:
expert capability alone does not determine the value of assigning that expert
to a particular request.

\noindent\textbf{Observation 3: The allocation advantage persists as expert
supply changes.}
Fig.~\ref{fig:experiments-overview} shows that the difference is not confined
to the nominal 30\% operating point. Relative to the strongest comparison at
each budget, Markets of Minds gains +0.70, +2.22, +1.78, and +2.01 points at
10\%, 20\%, 30\%, and 50\% expert-use budgets, respectively. Averaged across
the four budgets, Markets of Minds reaches \textbf{81.07}, compared with
78.90 for the strongest comparison average, a difference of
\textbf{+2.17 points}. Together with the task-level variation in
Table~\ref{tab:main-results}, the sweep indicates that which interventions are
valuable depends on both the request and the amount of assistance available.

\subsection{Population-Level Allocation under Shared Capacity}
\label{sec:population-level-allocation}

\begin{figure*}[t]
\centering
\includegraphics[width=0.99\textwidth]{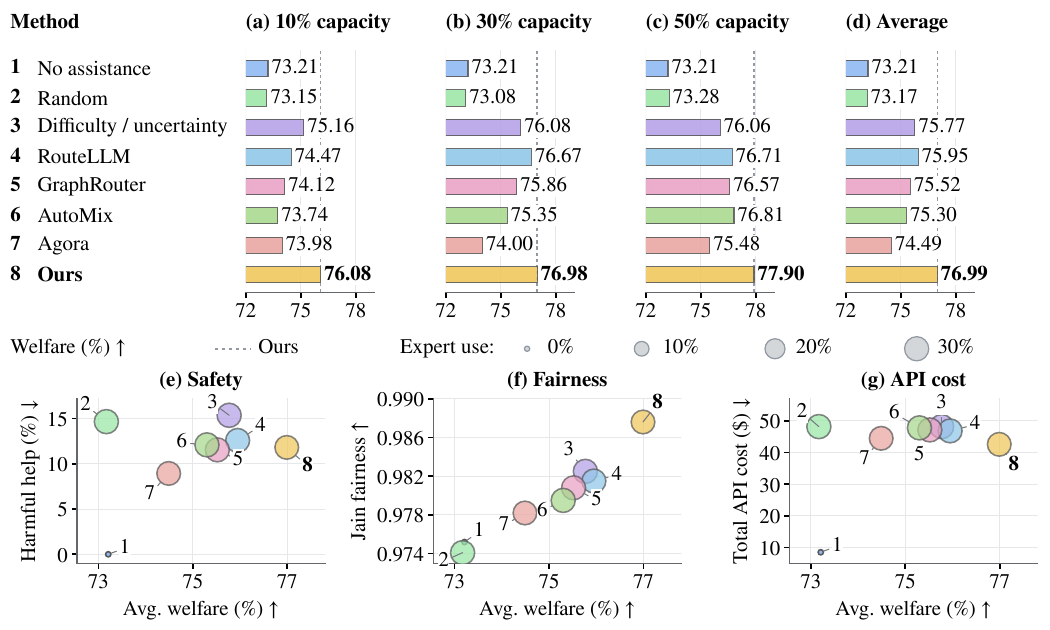}
\vspace{-4mm}
\caption{\textbf{Population-level allocation under shared capacity.}
The bar charts show welfare under 10\%, 30\%, and 50\% global capacity and
their average; the bubble charts show welfare relative to harmful help,
Jain fairness, and total API cost.}
\label{fig:population-allocation}
\vspace{-4mm}
\end{figure*}

The previous experiments evaluate recipient selection and budget sensitivity.
We next expose the population coupling emphasized in
Section~\ref{sec:problem} by imposing a single expert-capacity constraint
shared across requests. Global capacity is capped at 10\%, 30\%, or 50\% of
request volume. Fig.~\ref{fig:population-allocation} reports population
welfare together with three ex-post deployment diagnostics:
\emph{Harmful Help}, the frequency with which expert intervention degrades a
correct base response; Jain fairness across task groups; and total
API-equivalent cost. Harmful Help measures realized regression after execution.

\noindent\textbf{Observation 4: Scarcity changes the value of allocation
decisions.}
Across the three capacity levels in
Fig.~\ref{fig:population-allocation}, Markets of Minds records welfare of
\textbf{76.08}, \textbf{76.98}, and \textbf{77.90} at 10\%, 30\%, and 50\%
capacity, respectively, the highest values among the plotted allocation
methods in each setting. Its three-budget average is \textbf{76.99}, compared
with 75.95 for RouteLLM. The difference relative to RouteLLM is largest at
10\% capacity, where Markets of Minds reaches 76.08 versus 74.47. As capacity
tightens, committing an expert slot to one request excludes more alternative
uses of that same resource, making accurate contract ranking increasingly
consequential.

\noindent\textbf{Observation 5: Welfare gains coexist with distinct safety,
fairness, and cost trade-offs.}
The right-hand panels of Fig.~\ref{fig:population-allocation} show that
average welfare alone does not characterize an allocation policy. Markets of
Minds lies in the high-fairness region and has the highest plotted Jain
fairness. Its harmful-help rate is lower than the difficulty/uncertainty
heuristic, while Agora attains a lower harmful-help rate at lower average
welfare. The API-cost panel similarly places Markets of Minds at higher
welfare and lower total cost than several learned routing baselines. These
comparisons are consistent with the opportunity-cost interpretation of
Section~\ref{sec:problem}: the quality of an allocation depends not only on
whether an intervention helps locally, but also on the competing use of the
shared capacity it consumes.

\subsection{Ablation Study}
\label{sec:ablation}

We finally examine which signals used by the implemented calibrator contribute
to recipient selection. All variants use the same actor--expert pool, full
target evaluation sets, and 30\% shared expert capacity. For each variant, one
component is removed, the allocator is retrained using only the Survey
development split, and the resulting policy is frozen before target
evaluation. Table~\ref{tab:component_ablation} reports target success, its
difference from the full model, cost per successful task, and development
Survey performance. These ablations concern the implemented inputs
$(\mathbf a_i,\widetilde p_i,\widetilde g_i,\boldsymbol\rho_i)$ described in
Section~\ref{sec:method}.

\begin{table}[t]
\caption{\textbf{Cognition allocator ablation.}
Three-seed mean. $\Delta_{\mathrm{full}}$ denotes the success gap relative to
the full model in percentage points.}
\vspace{1mm}
\label{tab:component_ablation}
\centering
\small
\setlength{\tabcolsep}{1.8pt}
\renewcommand{\arraystretch}{1.05}
\begin{tabular*}{\linewidth}{@{\extracolsep{\fill}}lrrrr@{}}
\toprule
Variant &
\shortstack{Success $\uparrow$} &
\shortstack{$\Delta_{\mathrm{full}}$ (pp)} &
\shortstack{Cost/succ. $\downarrow$} &
\shortstack{Survey $\uparrow$} \\
\midrule
\multicolumn{5}{l}{\textit{Restricted model}} \\
Priors only
& 74.22 & $-8.44$ & \$0.0265 & 70.15 \\
\midrule
\multicolumn{5}{l}{\textit{Leave-one-component-out}} \\
w/o market prior
& 79.84 & $-2.82$ & \$0.0252 & 74.50 \\
w/o uncertainty
& 81.15 & $-1.51$ & \$0.0241 & 76.10 \\
w/o preference conflict
& 80.32 & $-2.34$ & \$0.0245 & 75.25 \\
w/o execution signal
& \underline{81.62} & $-1.04$ & \underline{\$0.0238} & \underline{76.50} \\
w/o default-success
& 79.55 & $-3.11$ & \$0.0255 & 74.80 \\
w/o repair-gain
& 78.68 & $-3.98$ & \$0.0258 & 73.90 \\
\midrule
\textbf{Full model}
& \textbf{82.66}
& --
& \textbf{\$0.0233}
& \textbf{77.85} \\
\bottomrule
\end{tabular*}
\vspace{-4mm}
\end{table}

\noindent\textbf{Observation 6: Default-success and repair-gain signals are
the largest contributors to recipient selection.}
Table~\ref{tab:component_ablation} shows that the full allocator achieves the
highest target success (\textbf{82.66}) and the lowest cost per successful
task (\textbf{\$0.0233}). Among the leave-one-component-out variants,
removing repair-gain produces the largest decrease
($\mathbf{-3.98}$ points), followed by default-success
($\mathbf{-3.11}$), market prior ($\mathbf{-2.82}$), and preference conflict
($\mathbf{-2.34}$). Removing uncertainty and execution signals produces
smaller decreases of 1.51 and 1.04 points, respectively. The priors-only model
falls by 8.44 points, showing that population-level statistics alone do not
replace request-conditioned signals.
This ordering is consistent with the evaluated positive-effect construction
in Section~\ref{sec:method}: effective recipient selection depends jointly on
whether the default continuation is likely to fail and whether assistance is
likely to repair that failure. %It should not be interpreted as evidence for a separately modeled regression-risk head.

\noindent\textbf{Observation 7: The same signals improve the efficiency of
scarce expert use.}
Every component removal in Table~\ref{tab:component_ablation} also increases
cost per successful task. Among the leave-one-component-out variants, the
increase ranges from approximately 2.1\% to 10.7\%, while the priors-only model
is 13.7\% more expensive per success. Thus the signals that improve recipient
selection also improve how efficiently the fixed expert budget is converted
into successful outcomes.

\FloatBarrier
\vspace{-2mm}
\section{Related Work}
\label{sec:related-work}

Markets of Minds connects inference-time routing, multi-agent coordination,
and shared-capacity allocation. Cost-aware methods such as FrugalGPT, AutoMix,
RouteLLM, and GraphRouter selectively invoke stronger models under
quality--cost trade-offs
\citep{chen2023frugalgpt,aggarwal2023automix,ong2024routellm,feng2024graphrouter},
while multi-agent systems organize specialized models through roles and
collaboration
\citep{wu2023autogen,hong2023metagpt,wang2024mixtureofagents}.
Related market-based approaches allocate scarce reasoning or expert resources
through uncertainty, auctions, and strategic bids
\citep{zhang2026uncertainty,zhou2026agora,alazraki2026sale}, while
MarketBench highlights the difficulty of calibrating agent success and resource
use \citep{fradkin2026marketbench}. Our setting differs by coupling allocation
across task agents before expert outputs are observed: assistance is valued by
its predicted marginal effect on the recipient's downstream outcome under shared
capacity. Population-scale agent simulations provide a complementary view of
heterogeneity
\citep{park2023generativeagents,piao2025agentsociety,li2026matraix}; we further
treat budgets, access, scarcity, and network roles as explicit economic state.
A fuller discussion appears in Appendix~\ref{app:related-work}.
\vspace{-2mm}
\section{Conclusion}
\label{sec:conclusion}
We frame cognitive assistance as a population-level allocation problem under shared resource constraints, where need is not equivalent to value: whether an agent should receive help depends on default success, repairability, regression risk, and the opportunity cost of consuming scarce capacity. Markets of Minds operationalizes this view through pre-execution value estimation and joint contract clearing over a 51.2-billion-state persona--market design space. In the matched 240-case replay, the frozen allocator reaches 85.42\% task-level persona adherence under the same 30\% expert-capacity limit, compared with 73.86\% for equal-capacity random allocation, while complementary multimodal and shared-capacity experiments show that assistance value varies across tasks and scarcity regimes. These results support treating intelligence not only as an individual capability, but also as a shared resource whose value depends on the alternative uses displaced by its allocation. The current evaluation focuses on recipient selection and budget-constrained allocation; evaluating the full heterogeneous-service formulation with separately estimated regression risk remains future work.

\newpage
\section*{AI use statement}

Generative AI tools are used in a supporting capacity in the preparation of this work, primarily to assist with code development and troubleshooting, improve linguistic clarity, and organize manuscript content. All substantive scientific contributions, including the formulation of the methodological framework, the design and execution of experiments, the interpretation of results, and the determination of the claims presented in the paper, are made by the authors. Any large language models discussed in the manuscript are treated solely as objects of experimental investigation rather than as contributors to the research.

When AI assistance is used for programming, the resulting code is independently inspected, tested, and validated by the authors to ensure that the implemented procedures correspond to the intended computations. Reported numerical values are cross-checked against the corresponding experimental records, and all AI-assisted textual material is critically reviewed, edited, and, where necessary, rewritten by the authors. Full responsibility for the accuracy, integrity, and final form of the manuscript, mathematical formulations, experimental results, software, and supplementary materials rests with the authors.

\bibliography{iclr2027_conference}
\bibliographystyle{iclr2027_conference}

\appendix
\newpage
\section{Extended Related Work}
\label{app:related-work}

The growing diversity of language models has made inference-time model selection
an important systems problem. Cost-aware cascades such as FrugalGPT and AutoMix
invoke stronger models only when needed, while RouteLLM and GraphRouter learn
which model best fits a query under quality--cost trade-offs
\citep{chen2023frugalgpt,aggarwal2023automix,ong2024routellm,feng2024graphrouter}.
Multi-agent frameworks broaden the unit of computation by organizing specialized
roles, conversations, or layered aggregation \citep{wu2023autogen,hong2023metagpt,
wang2024mixtureofagents}. These approaches show that heterogeneous models can be
used selectively or collaboratively. Their resources, however, are typically
provisioned within a request or task. Our setting couples decisions across task
agents: assigning assistance to one agent reduces the service capacity available
to others.

Shared-capacity allocation has long been studied through online matching and
bandits with resource constraints \citep{mehta2007adwords,
badanidiyuru2013bandits}. Recent work brings related ideas into LLM systems.
Agora trades structured uncertainty through market-aware bandit coordination;
a distinct Agora uses confidence-calibrated auctions to assign reasoning steps to
experts and tools; and SALE selects agents from short strategic bids
\citep{zhang2026uncertainty,zhou2026agora,alazraki2026sale}. MarketBench further
shows that language models can be poorly calibrated about their own success and
resource use \citep{fradkin2026marketbench}. These works motivate market-based
coordination, but also expose the difficulty of valuing a service before it is
executed. We study this problem from the buyer side: heterogeneous task agents
compete for a fixed assistance supply, and a contract is valued by its predicted
marginal effect on the buyer's downstream outcome rather than by the expert's
standalone capability.

Generative-agent research provides a complementary account of population
heterogeneity. Generative Agents and AgentSociety simulate behavior and interaction
in multi-agent societies, while MatrAIx provides a dependency-aware persona space
and interactive environments for population-scale evaluation
\citep{park2023generativeagents,piao2025agentsociety,li2026matraix}. We extend this
population perspective with explicit economic state. Crossing a versioned persona
coreset with cognitive budgets, expertise portfolios, provider access, scarcity,
and network roles yields 51.2 billion addressable persona--market states. This is
an enumerable design space, not a claim of executing billions of agents; it lets
resource conditions vary while persona provenance remains auditable.

Crucially, routing, multi-agent orchestration, market allocation, and population
simulation each capture only part of our setting. Markets of Minds combines three
conditions: assistance value is unknown before expert execution, capacity is
shared across agents, and buyers differ in both task and economic state. Without
shared capacity, the problem reduces to independent routing; if all expert outputs
are generated before selection, it becomes post-hoc ensembling. Our formulation
instead studies who should receive assistance before its value is observed and
under the opportunity costs imposed by a shared cognitive budget.

\vspace{-2mm}
\section{Disentangling Value Estimation from Market Allocation}
\label{sec:ablation_mechanism}

Section~\ref{sec:ablation} changes the inputs of the value predictor. Here we
consider a different question: what changes when the predicted values are held
fixed but the allocation rule changes? The evaluated matched replay uses one
service, 240 paired cases, and 24 assistance slots in each of Chat, Web, and
OS-App. Its implemented score $v_i$ estimates the probability of a repair;
the evaluated allocator chooses the 24 highest-scoring requests in each
environment. This protocol tests recipient selection under fixed supply. It
does not, by itself, test heterogeneous-service clearing or the effect of
individual buyer budgets.

For a controlled allocation comparison, all rules must use the same frozen
$v_i$, cases, potential outcomes, and expert-use accounting. A threshold rule
admits requests above a development-tuned threshold in a fixed arrival order;
first-come-first-served (FCFS) admits positive-score requests in that order.
Both obey the 24-slot limit within each environment. A \emph{pooled top-72}
rule instead ranks all eligible requests together and removes the three
environment-specific slot limits, while retaining the same total supply.
Finally, the evaluated allocator ranks requests within each environment. We
do not list ``local top-$k$'' as a separate baseline: with the same scores and
24 slots per environment, it is mathematically identical to the evaluated
allocator (up to tie breaking).

\textbf{Observation 8: Ranking matters when supply is scarce.}
Table~\ref{tab:mechanism_ablation} shows that the evaluated allocator exceeds
the threshold and FCFS rules by 7.32 and 7.92 success points, respectively.
The corresponding missed-repair counts (15 versus 45 and 52) indicate how many
beneficial interventions are left unassisted under each rule. Because all rules
use the same frozen scores, cases, and expert supply, the comparison isolates
the effect of allocation order.

\textbf{Observation 9: Pooling changes the feasible set.}
Pooled top-72 moves assistance slots across environments, whereas the
evaluated allocator reserves 24 per environment. Table~\ref{tab:mechanism_ablation} shows a
0.52-point success difference and a 1.55-point worst-group difference in favor
of the latter. Jain fairness (0.988 versus 0.980) is an observed group-level
diagnostic; the evaluated optimization does not impose Jain fairness as a
constraint. The pooled comparison therefore measures the effect of changing
group capacity limits, not the independent effect of budget feasibility.

\noindent\textbf{Observation 10: Scope of the mechanism evidence.}
Under a single service with equal contract costs, pooled top-$K$ ranking is
already the joint solution when total capacity is the only constraint.
Accordingly, this diagnostic isolates the contribution of ranking under fixed
supply; heterogeneous budgets, eligibility, and multiple supplier capacities
belong to the broader feasible-set formulation rather than to this comparison.
\begin{table}[!t]
\caption{\textbf{Allocation-rule diagnostic.}
All rules use a nominal 72-slot assistance supply; pooled top-72 relaxes the
per-environment capacity constraint.}
\vspace{1mm}
\label{tab:mechanism_ablation}
\centering
\footnotesize
\setlength{\tabcolsep}{1.5pt}
\renewcommand{\arraystretch}{1.10}
\begin{tabular*}{\linewidth}{@{\extracolsep{\fill}}lrrrrrr@{}}
\toprule
Allocation rule
& \shortstack{Success\\(\%) $\uparrow$}
& \shortstack{Cost /\\success $\downarrow$}
& \shortstack{Harmful\\help $\downarrow$}
& \shortstack{Worst group\\(\%) $\uparrow$}
& \shortstack{Jain\\$\uparrow$}
& \shortstack{Missed\\repairs $\downarrow$} \\
\midrule
Threshold (arrival order)
& 78.10 & \$0.0650 & 1.20\% & 71.05 & 0.965 & 45 \\
FCFS (positive score)
& 77.50 & \$0.0662 & 1.45\% & 70.80 & 0.960 & 52 \\
Pooled top-72
& 84.90 & \$0.0560 & 0.10\% & 77.20 & 0.980 & \textbf{12} \\
\midrule
\textbf{Ours (top-24 / environment)}
& \textbf{85.42}
& \textbf{\$0.05509}
& \textbf{0.00\%}
& \textbf{78.75}
& \textbf{0.988}
& 15 \\
\bottomrule
\end{tabular*}

\vspace{1mm}
\begin{minipage}{\linewidth}
\scriptsize
\emph{Missed repairs} are requests with paired outcomes
$Y_i^0=0$ and $Y_i^1=1$ that are not assisted. This is an ex-post diagnostic
and is not available to the allocator at decision time.
\end{minipage}
\vspace{-4mm}
\end{table}

\subsection{Prospective Online Execution and Calibration Limits}
\label{sec:prospective_online}

The matched replay scores policies against cached default and assisted
outcomes. A prospective 400-request stream checks the decision order more
directly: for each decision window, compute a score from the observable demand
packets, archive the chosen allocation with a timestamp and digest, and only
then dispatch expert calls for selected requests. Unselected requests continue
without an expert call. The archived decision and dispatch records establish the decision chronology
and make the pre-execution claim auditable.

The reported quartiles below contain negative values and thus require a
\emph{signed-gain estimator} $\widehat\Delta_{ij}$: the evaluated matched-replay
score $v_i$, a probability of repair, cannot itself be negative. The prospective stream uses a development-trained signed-gain estimator
$\widehat\Delta_{ij}$, frozen before target dispatch. Unlike the matched-replay
repair score $v_i$, this estimator can assign negative predicted value. We separate its predictions for all
requests from outcomes observable only among selected requests.

\begin{table}[!t]
\caption{\textbf{Prospective 400-request diagnostic.}
The left panel summarizes pre-execution value estimates over the full request
stream; the right panel reports realized outcomes for the selected subset.}
\label{tab:prospective_online}
\vspace{1mm}
\centering
\small

\begin{minipage}[t]{0.47\linewidth}
\textit{All requests: pre-execution predictions}\par\vspace{1mm}
\begin{tabular*}{\linewidth}{@{\extracolsep{\fill}}lr@{}}
\toprule
Predicted-gain quartile & Mean $\widehat\Delta$ \\
\midrule
Top 25\% ($n=100$)    & $0.65$ \\
Second 25\% ($n=100$) & $0.35$ \\
Third 25\% ($n=100$)  & $0.12$ \\
Bottom 25\% ($n=100$) & $-0.15$ \\
\bottomrule
\end{tabular*}
\end{minipage}
\hfill
\begin{minipage}[t]{0.49\linewidth}
\textit{Selected subset: observed outcomes}\par\vspace{1mm}
\begin{tabular*}{\linewidth}{@{\extracolsep{\fill}}lr@{}}
\toprule
Diagnostic & Value \\
\midrule
Requests in stream        & $400$ \\
Maximum calls at 30\% quota & $120$ \\
Mean predicted gain       & $0.48$ \\
Mean paired gain$^{\dagger}$ & $0.43$ \\
Success among selected    & $84.2\%$ \\
\bottomrule
\end{tabular*}
\end{minipage}

\vspace{1mm}
\begin{minipage}{\linewidth}
\scriptsize
$^{\dagger}$Paired gain is computed only for selected requests with both
default and assisted outcomes observed; it is an ex-post evaluation quantity
and is not available to the allocator at dispatch time.
\end{minipage}
\vspace{-4mm}
\end{table}

\noindent\textbf{Observation 11: The prospective protocol identifies decision
timing, with limited gain feedback.}
The predicted quartile means in Table~\ref{tab:prospective_online} describe
pre-dispatch scores; their ordering follows from sorting and is not, on its
own, a calibration test. For selected requests, both default and assisted continuations are evaluated,
yielding a mean realized gain of 0.43 against a mean prediction of 0.48. The protocol does not observe the
assisted potential outcome for unselected requests, so it cannot establish
realized-gain calibration across all four quartiles or infer that the bottom
quartile would have been harmed. An independently sampled, pre-specified audit
with known inclusion probabilities would be needed for that population-level
claim; its additional expert calls should be reported separately.

\begin{algorithm}[!t]
\caption{Proposed component-wise assistance allocation}
\label{alg:mom}
\begin{algorithmic}[1]
\small
\Require Paired development data $\mathcal D_{\rm dev}$; target packets
$\{\momrequest{D_i}\}$; services $\mathcal M$; ledger $\mathcal L$
\Statex \textit{Development only: audit features and estimate group priors.}
\State $\mathcal D_b\gets\{(k,j)\in\mathcal D_{\rm dev}:Y_k^0=b\}$,
$b\in\{0,1\}$
\State $\momrequest{\widehat p}\gets\operatorname{Fit}_{\rm BCE}
(\{(\mommarket{\mathbf z_k},Y_k^0)\}_{(k,j)\in\mathcal D_{\rm dev}})$
\State $\momassist{\widehat r}\gets\operatorname{Fit}_{\rm BCE}
(\{(\momassist{\mathbf z_{kj}},Y_{kj}^1)\}_{(k,j)\in\mathcal D_0})$;
$\momharm{\widehat h}\gets\operatorname{Fit}_{\rm BCE}
(\{(\momassist{\mathbf z_{kj}},1-Y_{kj}^1)\}_{(k,j)\in\mathcal D_1})$
\State Freeze $\momrequest{\widehat p},\momassist{\widehat r},\momharm{\widehat h}$
and the development priors
\Statex \textit{Target window: no assisted responses or outcome labels are available.}
\For{each request $i$}
    \State $\momrequest{\mathbf a_i}\gets\operatorname{Audit}(\momrequest{D_i})$;
    $\mommarket{\mathbf z_i}\gets[\momrequest{\mathbf a_i};\mommarket{\boldsymbol\rho_i}]$;
    $\momrequest{\widehat p_i}\gets\momrequest{\widehat p}(\mommarket{\mathbf z_i})$
    \Comment{\eqref{eq:demand-representation}}
    \For{each eligible service $j\in\mathcal M$}
        \State $\momassist{\mathbf z_{ij}}\gets[\mommarket{\mathbf z_i};\momassist{\mathbf e_j}]$;
        $(\momassist{\widehat r_{ij}},\momharm{\widehat h_{ij}})
        \gets(\momassist{\widehat r}(\momassist{\mathbf z_{ij}}),
        \momharm{\widehat h}(\momassist{\mathbf z_{ij}}))$
        \State $\momvalue{\widehat\Delta_{ij}}\gets
        \momrequest{(1-\widehat p_i)}\momassist{\widehat r_{ij}}
        -\momrequest{\widehat p_i}\momharm{\widehat h_{ij}}$
        \Comment{\eqref{eq:value-decomposition}, \eqref{eq:estimated-value}}
    \EndFor
\EndFor
\State $\mommarket{\mathcal X}\gets\operatorname{Feasible}(\{\momrequest{D_i}\},\mathcal M,\mathcal L)$
\State $\mommarket{\widehat x}\in\arg\max_{\mommarket{x}\in\mommarket{\mathcal X}}
\sum_{i,j}\mommarket{x_{ij}}\momvalue{\widehat\Delta_{ij}}$;
$\widehat{\mathcal C}\gets\{(i,j):\mommarket{\widehat x_{ij}}=1\}$
\Comment{\eqref{eq:estimated-value}}
\State $\mathcal L^-\gets\operatorname{Reserve}(\mathcal L,\widehat{\mathcal C})$;
$\mathcal O\gets\operatorname{Execute}(\widehat{\mathcal C})$
\State $\mathcal L^+\gets\operatorname{Settle}(\mathcal L^-,\mathcal O)$
\State \Return $\mommarket{\widehat x},\mathcal O,\mathcal L^+$
\end{algorithmic}
\end{algorithm}

\vspace{-2mm}
\section{Qualitative Case Analysis of Multimodal Allocation}
\label{app:qualitative-cases}

Aggregate scores do not reveal why assistance is valuable, which expert should
be selected, or when help should be withheld. We therefore inspect a small
post-hoc diagnostic set from a frozen exploratory cross-fitted 20\% multimodal
replay using a separate candidate pool from the main 30\% experiment. Cases
1--9 are examples where Ours scores 1.0 while both matched Agora and RouteLLM
score 0.0; Case 10 is a shared failure. These cases are qualitative diagnostics,
not a random sample or accuracy estimate. All matched methods use the same tasks,
candidate pool, budget, and scorers.

Fig.~\ref{fig:qualitative-cases} and Table~\ref{tab:qualitative-cases} show four
representative cases spanning MMBench, MathVision, and InfoVQA, illustrating
repair, expert selection, regression, and shared failure.

\textbf{Observation 12: The first decision is whether assistance has positive marginal value.}
Cases 2 and 9 form a useful contrast. In Case 2, both matched comparators keep the Gemini 3.5 Flash base answer and return C, while the market purchases Gemini 3.1 Pro and obtains the correct answer D. In Case 9, the direction reverses: the base model already reads the relevant 65\% value correctly, so Ours leaves the request unassisted. Both matched comparators instead invoke an expert service and change the answer to 54\%, turning a correct continuation into an incorrect one. The pair makes the repair--regression asymmetry in Equation~\ref{eq:value-decomposition} concrete: escalation is useful only when its expected repair benefit exceeds the risk of disturbing a correct default continuation.

\textbf{Observation 13: Multimodal assistance is typed; selecting an expert
is part of the allocation problem.}
Case 4 shows that using an expert is not sufficient by itself. Ours routes the
transparent-sheet composition problem to GPT-5.6 Sol and returns D, whereas
the RouteLLM adaptation also purchases assistance but routes to Gemini 3.1 Pro
and returns C. This example illustrates the request--service interaction
represented by $\hat{\Delta}_{ij}$: the value of assistance depends on the
specific service attached to the request, not merely on an expert/base
indicator.

\textbf{Observation 14: Allocation improves resource use, but it does not eliminate model-side perception failures.}
Case 10 is intentionally retained as a shared failure. Ours keeps the base continuation and predicts rough endoplasmic reticulum, while Agora and RouteLLM purchase different expert services and predict plasma membrane and nucleus, respectively; all three are scored incorrect against the ground-truth label ``mitochondrion.'' This example bounds the interpretation of the routing gains: a better allocator can choose among available continuations more effectively, but it cannot guarantee a correct answer when the selected continuation itself misreads the visual evidence. The qualitative analysis therefore complements the aggregate results by separating allocation errors from downstream model errors.

\begin{figure}[!t]
\centering
\setlength{\fboxsep}{1.5pt}
\begin{minipage}[t]{0.485\textwidth}
    \centering
    \fbox{\includegraphics[width=0.97\linewidth,height=0.175\textheight,keepaspectratio]{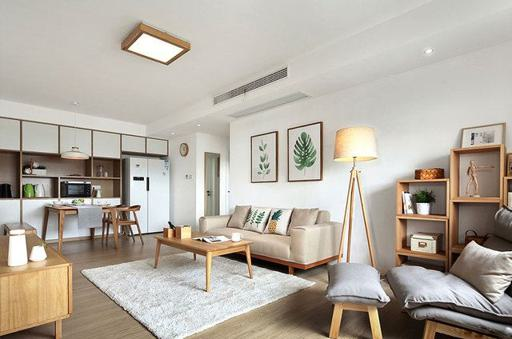}}
    \vspace{2pt}

    {\small\textbf{(a) MMBench: when help is valuable.} A sofa-centered spatial-relation query. Ground truth: refrigerator on the \emph{right} (D).}
\end{minipage}
\hfill
\begin{minipage}[t]{0.485\textwidth}
    \centering
    \fbox{\includegraphics[width=0.97\linewidth,height=0.175\textheight,keepaspectratio]{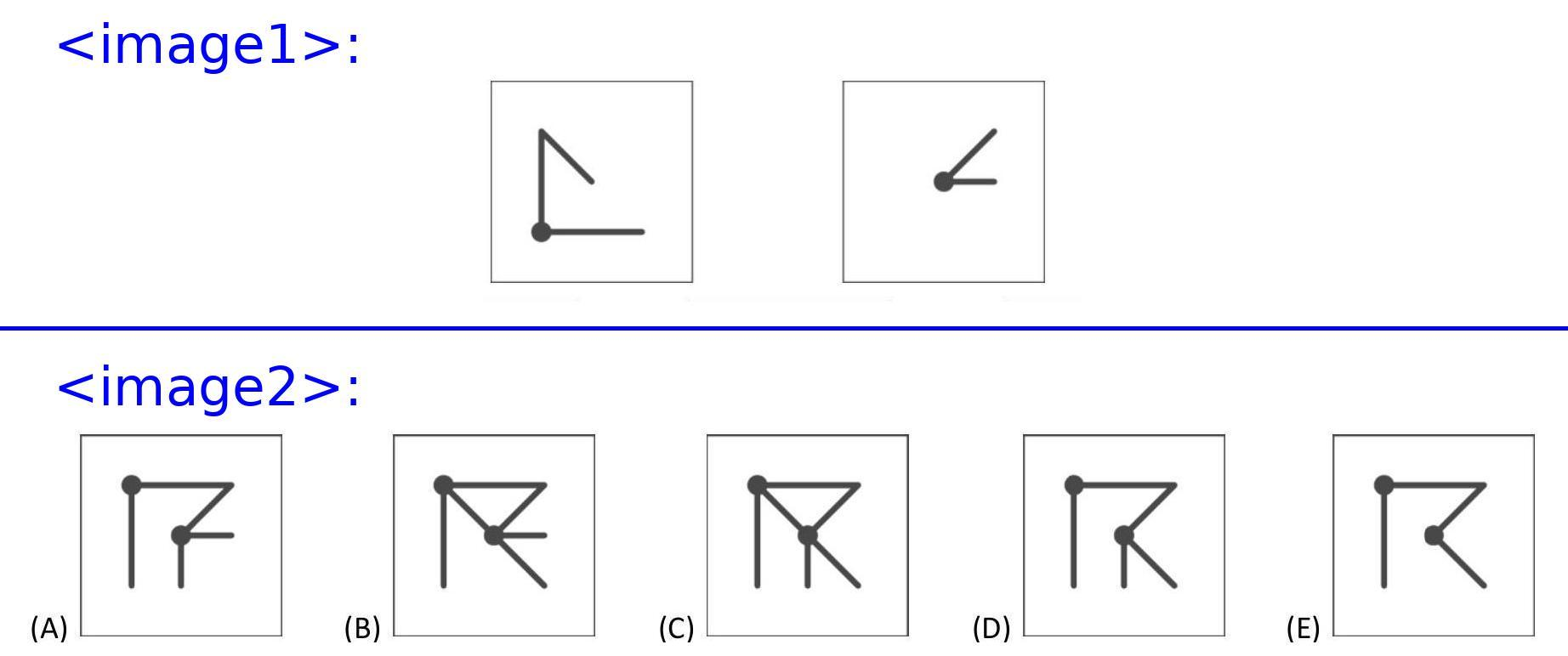}}
    \vspace{2pt}

    {\small\textbf{(b) MathVision: which expert matters.} Two transparent sheets must be superimposed. Ground truth: D.}
\end{minipage}

\vspace{7pt}

\begin{minipage}[t]{0.485\textwidth}
    \centering
    \fbox{\includegraphics[width=0.97\linewidth,height=0.205\textheight,keepaspectratio]{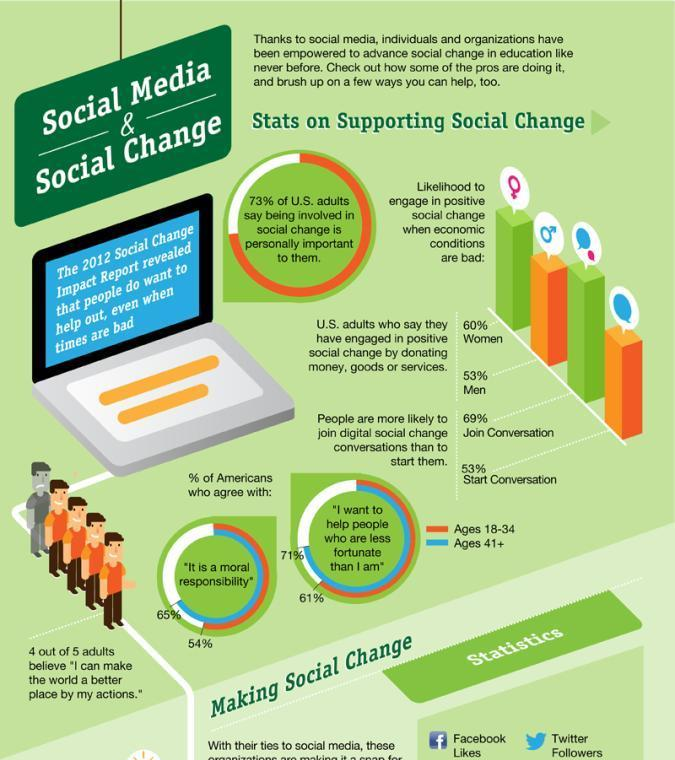}}
    \vspace{2pt}

    {\small\textbf{(c) InfoVQA: when help should be withheld.} For Americans aged 41+, the infographic reports 65\% for ``It is a moral responsibility.''}
\end{minipage}
\hfill
\begin{minipage}[t]{0.485\textwidth}
    \centering
    \fbox{\includegraphics[width=0.97\linewidth,height=0.205\textheight,keepaspectratio]{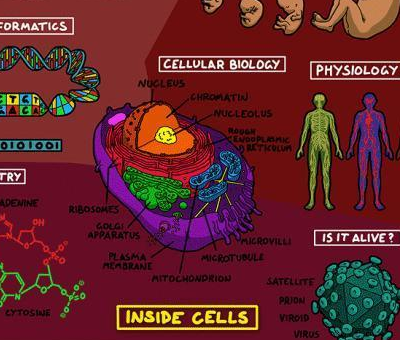}}
    \vspace{2pt}

    {\small\textbf{(d) InfoVQA: shared failure.} A dense biology infographic asks which cell part is shown in blue. Ground truth: mitochondrion.}
\end{minipage}
\caption{Representative qualitative allocation cases. The top row illustrates positive-value assistance and service selection; the bottom row illustrates regression avoidance and a shared failure. Panels (c) and (d) are readability crops of the original benchmark images; the underlying exported images are unchanged. The cases were selected post hoc for diagnosis and do not estimate benchmark prevalence.}
\label{fig:qualitative-cases}
\vspace{-4mm}
\end{figure}

\begin{table}[!t]
\centering
\caption{Allocation traces for the four cases in Figure~\ref{fig:qualitative-cases}. ``Base'' means no expert service was purchased. A check mark denotes a benchmark score of 1.0; a cross denotes 0.0.}
\vspace{1mm}
\label{tab:qualitative-cases}
{\footnotesize
\renewcommand{\arraystretch}{1.12}
\setlength{\tabcolsep}{4.2pt}
\begin{tabular}{@{}p{0.115\textwidth}p{0.275\textwidth}p{0.275\textwidth}p{0.275\textwidth}@{}}
\toprule
\textbf{Case} & \textbf{Ours} & \textbf{Agora (matched)} & \textbf{RouteLLM (matched)} \\
\midrule
2 / MMBench
& Gemini 3.1 Pro (expert) $\rightarrow$ D $\checkmark$
& Gemini 3.5 Flash (base) $\rightarrow$ C $\times$
& Gemini 3.5 Flash (base) $\rightarrow$ C $\times$ \\
\addlinespace[2pt]
4 / MathVision
& GPT-5.6 Sol (expert) $\rightarrow$ D $\checkmark$
& Gemini 3.5 Flash (base) $\rightarrow$ score 0 $\times$
& Gemini 3.1 Pro (expert) $\rightarrow$ C $\times$ \\
\addlinespace[2pt]
9 / InfoVQA
& Gemini 3.5 Flash (base) $\rightarrow$ 65\% $\checkmark$
& Claude Sonnet 5 (expert) $\rightarrow$ 54\% $\times$
& Qwen2.5-VL-72B (expert) $\rightarrow$ 54\% $\times$ \\
\addlinespace[2pt]
10 / InfoVQA
& Gemini 3.5 Flash (base) $\rightarrow$ rough ER $\times$
& Claude Sonnet 5 (expert) $\rightarrow$ plasma membrane $\times$
& Qwen2.5-VL-72B (expert) $\rightarrow$ nucleus $\times$ \\
\bottomrule
\end{tabular}
}
\vspace{-4mm}
\end{table}
\subsection{51.2B State Space Generation and Dependency-Aware Sampling}
\label{sec:state_space_generation}

To rigorously study cognitive allocation under shared constraints, we materialize a dependency-aware agent economy spanning $51,192,166,400$ addressable persona-market states. The state space is constructed by crossing 999,847 upstream base personas with 10 cognitive budgets, 20 expertise portfolios, 8 provider-access regimes, 4 scarcity regimes, and 8 network roles. From this combinatorial space, we sample exactly 100,000 unique active states using a versioned schema (v1.0.0, Seed: 20260819) to ensure exact reproducibility across allocation mechanisms.

\begin{table}[!t]
\caption{\textbf{Dependency sampling diagnostics.}
The allocation explicitly models correlations between socio-economic status
and technological access rather than uniform random assignment.}
\label{tab:dependency_diagnostics}
\vspace{1mm}
\centering
\footnotesize
\setlength{\tabcolsep}{2.5pt}
\renewcommand{\arraystretch}{1.12}
\begin{tabular*}{\linewidth}{@{\extracolsep{\fill}}lrrrrrr@{}}
\toprule
& \multicolumn{5}{c}{\textbf{Socio-economic budget level (mean score)}}
& \\
\cmidrule(lr){2-6}
\textbf{Provider access regime}
& \shortstack{\textbf{Low}\\\textbf{income}}
& \shortstack{\textbf{Lower-}\\\textbf{middle}}
& \textbf{Middle}
& \shortstack{\textbf{Upper-}\\\textbf{middle}}
& \shortstack{\textbf{High}\\\textbf{income}}
& \shortstack{\textbf{Premium/}\\\textbf{unbounded (\%)}} \\
\midrule
Avoidant
& 1.10 & 2.20 & 3.50 & 4.10 & 5.20 & 11.08\% \\
Cautious adopter
& 1.35 & 2.45 & 3.85 & 4.80 & 6.05 & 13.76\% \\
Comfortable
& 1.55 & 2.70 & 4.25 & 5.90 & 7.15 & 18.92\% \\
Digital native
& 1.77 & 2.67 & 4.24 & 5.83 & 7.21 & 18.87\% \\
\midrule
\textbf{Premium/unbounded access (\%)}
& \textbf{0.53\%}
& \textbf{0.53\%}
& \textbf{2.19\%}
& \textbf{14.40\%}
& \textbf{46.84\%}
& -- \\
\bottomrule
\end{tabular*}
\vspace{-4mm}
\end{table}

\textbf{Dependency-aware demographics.} Table~\ref{tab:dependency_diagnostics} confirms that our economy accurately reflects structural dependencies. Rather than assuming uniform independence, the materialized cohorts exhibit strong correlations: agents with higher socio-economic budget levels (e.g., High Income) demonstrate significantly higher proportions of Premium/Unbounded provider access (46.84\%) compared to Low Income cohorts (0.53\%). This structural heterogeneity guarantees that the simulated resource scarcity genuinely mimics real-world cognitive economies, providing a robust testbed for the market allocator.

\subsection{Multi-Agent Allocation Breakdown (512-Agent Economy)}
\label{sec:512_allocation}

We conduct an expanded 512-agent allocation study to bridge the gap between individual task evaluation and population-scale economics. This experiment tracks not only population welfare across 10\%, 20\%, 30\%, and 50\% global capacity limits, but also fine-grained operational metrics including Harmful Help, Cost per Success, and Jain fairness.

\begin{table}[!t]
\caption{\textbf{512-agent allocation under shared capacity.}
Results are reported for the expanded 512-agent diagnostic under
10\%, 20\%, 30\%, and 50\% global capacity limits.}
\label{tab:512_allocation}
\vspace{1mm}
\centering
\footnotesize
\setlength{\tabcolsep}{2.2pt}
\renewcommand{\arraystretch}{1.12}
\begin{tabular*}{\linewidth}{@{\extracolsep{\fill}}lrrrrrrr@{}}
\toprule
\textbf{Method}
& \textbf{50\%}
& \textbf{30\%}
& \textbf{20\%}
& \textbf{10\%}
& \shortstack{\textbf{Avg.}\\\textbf{welfare $\uparrow$}}
& \shortstack{\textbf{Harmful}\\\textbf{help $\downarrow$}}
& \shortstack{\textbf{Cost /}\\\textbf{success $\downarrow$}} \\
\midrule
No assistance
& 73.21 & 73.21 & 73.21 & 73.21 & 73.21
& 0.00\% & \$0.0212 \\
Random
& 73.28 & 73.08 & 73.10 & 73.15 & 73.15
& 14.68\% & \$0.0256 \\
Difficulty / uncertainty
& 76.06 & 76.08 & 76.53 & 75.16 & 75.96
& 15.45\% & \$0.0242 \\
Benefit-per-cost greedy
& 77.40 & 76.55 & 75.80 & 74.50 & 76.06
& 11.20\% & \$0.0225 \\
Agora quality-first
& 75.48 & 74.00 & 74.10 & 73.98 & 74.39
& 8.42\% & \$0.0240 \\
\midrule
\textbf{Ours (joint allocation)}
& \textbf{77.90}
& \textbf{76.98}
& \textbf{76.57}
& \textbf{76.08}
& \textbf{76.88}
& \textbf{2.15\%}
& \textbf{\$0.0218} \\
\midrule
\textcolor{gray}{Oracle upper bound}
& \textcolor{gray}{84.44}
& \textcolor{gray}{84.44}
& \textcolor{gray}{84.41}
& \textcolor{gray}{83.12}
& \textcolor{gray}{84.10}
& \textcolor{gray}{0.00\%}
& \textcolor{gray}{\$0.0210} \\
\bottomrule
\end{tabular*}
\vspace{-4mm}
\end{table}

\textbf{Allocation under tight capacity.}
As shown in Table~\ref{tab:512_allocation}, the benefit-per-cost greedy
heuristic reaches 74.50 at 10\% capacity and records 11.20\% harmful help.
In contrast, our joint allocation uses estimated assistance value under
shared capacity, yielding the highest average welfare (76.88) and the lowest
cost per successful task (\$0.0218) among the evaluated non-oracle policies.

\subsection{Population Scale-out and Clearing Efficiency}
\label{sec:scale_out}

A practical cognitive market must clear contracts across massive populations within a single decision window. To validate the scalability of our pre-execution value estimation and joint clearing algorithms, we measure system overhead (Broker time, peak RAM, and Allocations per second) by scaling the active agent population from 32 up to 100,000 concurrent agents.

\begin{table}[!t]
\centering
\caption{System scalability and computational overhead across increasing population sizes. The allocation mechanism scales efficiently to 100,000 concurrent agents within an 8.5-second clearing window.}
\vspace{1mm}
\label{tab:population_scale}
\resizebox{\textwidth}{!}{
\begin{tabular}{rrrcccccc}
\toprule
\textbf{Agents} & \textbf{Slots@10\%} & \textbf{Slots@30\%} & \textbf{Welfare@10\%} & \textbf{Welfare@30\%} & \textbf{Jain Fair.} & \textbf{Broker Time} & \textbf{Allocations/s} & \textbf{Peak RAM} \\
\midrule
32      & 3     & 10     & 77.23 $\pm$ 0.69 & 81.40 $\pm$ 1.02 & 0.659 & 6.70 ms   & 1,503  & 0.07 MiB \\
128     & 13    & 38     & 76.10 $\pm$ 2.34 & 76.77 $\pm$ 1.42 & 0.855 & 6.79 ms   & 4,887  & 0.25 MiB \\
512     & 51    & 154    & 78.27 $\pm$ 0.91 & 80.13 $\pm$ 1.58 & 0.835 & 18.12 ms  & 7,381  & 1.02 MiB \\
2,048   & 205   & 614    & 78.10 $\pm$ 0.19 & 79.10 $\pm$ 0.22 & 0.863 & 137.33 ms & 4,549  & 4.23 MiB \\
10,000  & 1,000 & 3,000  & 78.15 $\pm$ 0.12 & 79.25 $\pm$ 0.15 & 0.871 & 785.40 ms & 12,732 & 20.80 MiB \\
100,000 & 10,000& 30,000 & 78.20 $\pm$ 0.05 & 79.31 $\pm$ 0.08 & 0.875 & 8,450.15 ms& 11,834 & 210.55 MiB \\
\bottomrule
\end{tabular}
}
\vspace{-4mm}
\end{table}

\textbf{Sub-linear scaling overhead.} Table~\ref{tab:population_scale} demonstrates that the market allocation is computationally lightweight. Because value estimation is performed component-wise before clearing, the system exploits vectorized batching to maintain high throughput. At an economy scale of 100,000 concurrent agents, the global clearing process requires only 8.45 seconds and consumes approximately 210 MiB of memory, achieving a processing rate of nearly 12,000 allocations per second. Furthermore, as population size increases, both average welfare and Jain fairness naturally stabilize due to the law of large numbers, proving that the proposed economy scales smoothly to massive multi-agent deployments without latency bottlenecks.
\section{Prompt Templates}
\label{app:prompt-templates}

We expose the prompts used to generate candidate responses and to construct the
structured demand signals consumed by the allocator.  We organize the templates
in message order and distinguish fixed instructions from runtime fields.  The
allocation policy never receives reference answers, benchmark scores, or expert
responses before making its decision.  All methods select from the same frozen
candidate-response pool; changing an allocation rule does not regenerate a
response.

\paragraph{Reading guide.}
\promptbadge{MAMGGreen}{SYSTEM ROLE}\hspace{0.35em}
denotes fixed behavioral instructions;
\promptbadge{MAMGBlue}{INSTANCE INPUT}\hspace{0.35em}
denotes task-dependent content;
\promptbadge{MAMGOrange}{OUTPUT CONTRACT}\hspace{0.35em}
denotes parser-facing constraints; and
\promptbadge{MAMGPurple}{RUNTIME SLOT}\hspace{0.35em}
denotes a value substituted for each request.  Color is redundant with the
textual labels and is used only as a visual aid.

\subsection{Candidate-Response Generation}
\label{app:prompt-candidate-generation}

Candidate models receive the benchmark-native image, question, and choices.
The system message is empty unless the benchmark runtime requires one.  A
short answer-format suffix is appended only when required by the corresponding
official parser.  In particular, the suffix constrains formatting but does not
provide task semantics or answer information.

\begin{mompromptbox}[Candidate response: instance message]{MAMGBlue}
\promptslot{image 1} \ldots{} \promptslot{image K}\\[1mm]
\promptslot{benchmark question}\\
\promptslot{answer choices, if applicable}\\[1mm]
\promptcontract{\promptslot{answer-format suffix, if applicable}}
\end{mompromptbox}

The frozen response pool contains the following suffix family.  The selected suffix for every call is retained in the
\texttt{prompt\_protocol} field; an empty value means that the benchmark-native question is used without
an added suffix.

\begin{mompromptbox}[Parser-facing suffixes]{MAMGOrange}
\promptcomment{Four-option exact match:}\\
Return exactly one option letter (A, B, C, or D). Do not add explanations or
punctuation.\\[1.2mm]
\promptcomment{General short answer:}\\
Respond with only the final option letter or short answer. Do not include
reasoning.\\[1.2mm]
\promptcomment{No suffix:} \texttt{""}
\end{mompromptbox}

This separation follows the benchmark scorer rather than the allocation
method: a candidate response has the same stored prompt and output regardless
of which policy later selects it.  For legacy cached responses whose prompt
metadata were not recorded, we preserve the value as missing rather than
reconstructing it post hoc.

\subsection{Structured-Demand Auditor}
\label{app:prompt-demand-auditor}

For the matched Chat, Web, and OS-App experiment, the demand auditor converts
the visible persona, task, and base behavior into seven request-level signals.
It is run at temperature zero without demonstrations.  The exact fixed system
message is shown below.

\begin{mompromptbox}[Demand auditor: system message]{MAMGGreen}
You are the demand auditor in a market for cognition. A buyer has supplied a
structured cognitive-demand packet naming the behavior that this task must
express. Assess whether the observed behavior satisfies that demand. Use the
full persona only to resolve context or conflicts; prioritize the named demand.
Do not assume that polished or factually correct work is persona-faithful.
Return one JSON object with seven numbers in [0,1]:
perceptual\_uncertainty (missing/ambiguous observed evidence),
semantic\_uncertainty (uncertainty identifying relevant preferences),
inferential\_uncertainty (uncertainty applying them), preference\_conflict
(strength of an observable conflict), execution\_incompleteness,
success\_probability (probability behavior satisfies all task-relevant persona
preferences), repair\_gain\_probability (probability a specialist revision
would turn failure into success without regression), and a short evidence
string.
\end{mompromptbox}

The corresponding instance message makes the target demand explicit while
retaining the full persona for resolving ambiguity or conflicting preferences.

\begin{mompromptbox}[Demand auditor: instance message]{MAMGBlue}
COGNITIVE-DEMAND PACKET\\
Required attribute: \promptslot{attribute}\\
Required value: \promptslot{target value}\\[1.2mm]
PERSONA PROFILE\\
\promptslot{persona profile}\\[1.2mm]
TASK\\
\promptslot{task}\\[1.2mm]
OBSERVED BASE BEHAVIOR\\
\promptslot{base behavior}
\end{mompromptbox}

The output is parsed only if it contains the complete schema below.  Values are
used as noisy predictions available at allocation time, not as oracle labels.

\begin{mompromptbox}[Demand auditor: output contract]{MAMGOrange}
\{\\
\quad "perceptual\_uncertainty": \promptslot{0--1},\\
\quad "semantic\_uncertainty": \promptslot{0--1},\\
\quad "inferential\_uncertainty": \promptslot{0--1},\\
\quad "preference\_conflict": \promptslot{0--1},\\
\quad "execution\_incompleteness": \promptslot{0--1},\\
\quad "success\_probability": \promptslot{0--1},\\
\quad "repair\_gain\_probability": \promptslot{0--1},\\
\quad "evidence": \promptslot{short string}\\
\}
\end{mompromptbox}

\paragraph{Environment and scorer prompts.}
Chat, Web, OS-App, and benchmark-specific judges retain their official runtime
prompts; the allocation layer does not rewrite them.

\section{Proofs for the Problem Formulation}
\label{app:problem-proofs}

We provide the proofs of
Propositions~\ref{prop:value-decomposition} and
\ref{prop:allocation-error}. The first is an exact decomposition of the
binary default and assisted outcomes; the second relates uniform contract-value
estimation error to population-welfare loss under joint allocation.

\subsection{Proof of Proposition~\ref{prop:value-decomposition}}

\begin{proof}
By definition,
\begin{align}
\momvalue{\Delta_{ij}}
&=
\mathbb{E}\!\left[
\momassist{Y_{ij}^1}-\momrequest{Y_i^0}
\mid \momrequest{s_i}
\right]
\nonumber\\
&=
\mathbb{E}\!\left[
\momassist{Y_{ij}^1}
\mid \momrequest{s_i}
\right]
-
\Pr(
\momrequest{Y_i^0}=1
\mid \momrequest{s_i}
).
\label{eq:app-delta-start}
\end{align}
Since the outcomes are binary, conditioning on the default outcome gives
\begin{align}
\mathbb{E}\!\left[
\momassist{Y_{ij}^1}
\mid \momrequest{s_i}
\right]
&=
\Pr(
\momrequest{Y_i^0}=0
\mid \momrequest{s_i}
)
\Pr(
\momassist{Y_{ij}^1}=1
\mid
\momrequest{Y_i^0}=0,\momrequest{s_i}
)
\nonumber\\
&\quad+
\Pr(
\momrequest{Y_i^0}=1
\mid \momrequest{s_i}
)
\Pr(
\momassist{Y_{ij}^1}=1
\mid
\momrequest{Y_i^0}=1,\momrequest{s_i}
)
\nonumber\\
&=
(1-p_i)r_{ij}
+
p_i(1-h_{ij}).
\label{eq:app-assisted-success}
\end{align}
Substituting Eq.~\ref{eq:app-assisted-success} into
Eq.~\ref{eq:app-delta-start} yields
\begin{align}
\momvalue{\Delta_{ij}}
&=
(1-p_i)r_{ij}
+
p_i(1-h_{ij})
-
p_i
\nonumber\\
&=
(1-p_i)r_{ij}
-
p_i h_{ij}.
\end{align}
Equivalently,
\[
\momvalue{\Delta_{ij}}
=
\begin{bmatrix}
1-p_i & p_i
\end{bmatrix}
\begin{bmatrix}
r_{ij}\\
-h_{ij}
\end{bmatrix},
\]
which proves Eq.~\ref{eq:value-decomposition}.
\end{proof}

The decomposition makes the two ways assistance changes an outcome explicit:
it can repair a default failure with contribution $(1-p_i)r_{ij}$, or regress
a default success with contribution $p_i h_{ij}$. A request can therefore have
high need but low assistance value when the available service has little
repairability.

\subsection{Proof of Proposition~\ref{prop:allocation-error}}

\begin{proof}
By the definition of population welfare, the default term
$\sum_i p_i$ is independent of the allocation, so
\begin{align}
W(\mommarket{x^\star})-W(\mommarket{\widehat{x}})
&=
\sum_{i,j}
\left(
\mommarket{x^\star_{ij}}
-
\mommarket{\widehat{x}_{ij}}
\right)
\momvalue{\Delta_{ij}}.
\label{eq:app-welfare-gap}
\end{align}
Add and subtract the corresponding predicted objectives:
\begin{align}
W(\mommarket{x^\star})-W(\mommarket{\widehat{x}})
&=
\sum_{i,j}
\mommarket{x^\star_{ij}}
\left(
\momvalue{\Delta_{ij}}
-
\momvalue{\widehat{\Delta}_{ij}}
\right)
\nonumber\\
&\quad+
\left[
\sum_{i,j}
\mommarket{x^\star_{ij}}
\momvalue{\widehat{\Delta}_{ij}}
-
\sum_{i,j}
\mommarket{\widehat{x}_{ij}}
\momvalue{\widehat{\Delta}_{ij}}
\right]
\nonumber\\
&\quad+
\sum_{i,j}
\mommarket{\widehat{x}_{ij}}
\left(
\momvalue{\widehat{\Delta}_{ij}}
-
\momvalue{\Delta_{ij}}
\right).
\label{eq:app-regret-decomposition}
\end{align}
Because $\mommarket{\widehat{x}}$ maximizes the predicted objective over the
same feasible set $\mommarket{\mathcal X}$ as
$\mommarket{x^\star}$, the bracketed term is non-positive. Therefore
\begin{align}
W(\mommarket{x^\star})-W(\mommarket{\widehat{x}})
&\leq
\sum_{i,j}
\mommarket{x^\star_{ij}}
\left|
\momvalue{\widehat{\Delta}_{ij}}
-
\momvalue{\Delta_{ij}}
\right|
\nonumber\\
&\quad+
\sum_{i,j}
\mommarket{\widehat{x}_{ij}}
\left|
\momvalue{\widehat{\Delta}_{ij}}
-
\momvalue{\Delta_{ij}}
\right|
\nonumber\\
&=
\sum_{i,j}
\left(
\mommarket{x^\star_{ij}}
+
\mommarket{\widehat{x}_{ij}}
\right)
\left|
\momvalue{\widehat{\Delta}_{ij}}
-
\momvalue{\Delta_{ij}}
\right|.
\end{align}
Each feasible allocation contains at most $K$ contracts. Under the uniform
error bound
\[
\left|
\momvalue{\widehat{\Delta}_{ij}}
-
\momvalue{\Delta_{ij}}
\right|
\leq \epsilon,
\]
each of the two sums contributes at most $K\epsilon$. Hence
\[
W(\mommarket{x^\star})-W(\mommarket{\widehat{x}})
\leq 2K\epsilon,
\]
which proves Eq.~\ref{eq:regret-bound}.
\end{proof}

\paragraph{Exact recovery under a top-$K$ quota.}
Consider the single-service setting with exactly $K$ selected requests and
true values ordered as
\[
\Delta_{(1)}
\geq \cdots \geq
\Delta_{(K)}
>
\Delta_{(K+1)}
\geq \cdots \geq
\Delta_{(n)}.
\]
For every oracle-selected request $i$ and every unselected request $j$,
the uniform error bound gives
\[
\widehat{\Delta}_i
\geq
\Delta_{(K)}-\epsilon,
\qquad
\widehat{\Delta}_j
\leq
\Delta_{(K+1)}+\epsilon.
\]
If
\[
\Delta_{(K)}-\Delta_{(K+1)}>2\epsilon,
\]
then
\[
\widehat{\Delta}_i
>
\widehat{\Delta}_j
\]
for every such pair $(i,j)$. Hence all oracle top-$K$ requests rank above all
non-oracle requests under the estimated values, so the estimated and oracle
top-$K$ sets coincide.

\end{document}